\documentclass[a4paper]{amsart}

\usepackage{a4wide}
\usepackage[foot]{amsaddr}
\usepackage[utf8]{inputenc}
\usepackage[T1]{fontenc}
\usepackage{amsmath}
\usepackage{amssymb,empheq,mathtools,dsfont}
\usepackage{physics}
\usepackage{stmaryrd}
\usepackage[mathscr]{eucal}
\usepackage{color}
\usepackage{nicefrac}
\usepackage{thm-restate}
\usepackage{calrsfs}
\usepackage{enumitem}

\DeclareMathAlphabet{\pazocal}{OMS}{zplm}{m}{n}

\usepackage{caption}
\usepackage{subcaption}
\usepackage[breaklinks,colorlinks=true, citecolor=blue]{hyperref}

\newcommand{\ie}{i.e.}

\newcommand*{\eqdef}{\stackrel{\text{def}}{=}}

\makeatletter
\newcommand*\defraccourci[3]{
  \expandafter\newcommand\csname #1#3\endcsname[1][]{#2{#1}}}
\newcommand*\generate[3]{
  \@for\@i:=#1\do{\expandafter\defraccourci\expandafter{\@i}{#2}{#3}}}
\makeatother

\generate{A,B,G,H,I,P,S,T,U,V,W,X}{\mathcal}{c} \generate{D,E}{\pazocal}{c} 

\generate{C,F,N,R,Z}{\mathbb}{}
\newcommand{\Fq}{\F_q}

\generate{a,b,c,q,r,s,t,u,v,w,x,y,z}{\mathbf}{v}
\renewcommand{\vec}[1]{\mathbf{#1}}

\renewcommand{\ev}{\vec{e}}
\renewcommand{\pv}{\vec{p}}
\newcommand{\unv}{\vec{1}}
\newcommand{\zerov}{\vec{0}}

\newcommand{\und}{\mathds{1}}

\generate{A,B,C,D,E,F,G,H,M,N,P,R,S,T,U,X,Y,Z}{\mathbf}{m}
\renewcommand{\Im}{\mathbf{I}}

\newcommand{\CC}{\ensuremath{\mathscr{C}}}
\newcommand{\EE}{\ensuremath{\mathscr{E}}}

\newcommand{\esp}{\mathbb{E}}
\newcommand{\Prob}{\mathbb{P}}
\newcommand{\prob}{\pi_{\ev}}
\newcommand{\Var}{\mathbf{Var}}

\newcommand{\sumV}{\sum_{\substack{V \le \F_q^n \\ \dim V = t}}}
\newcommand{\sumW}{\sum_{\substack{W \le \F_q^n \\ \dim W = t \\ W \neq V}}}
\newcommand{\sumFq}[1]{\sum_{#1 \in \F_q^\ast}}
\newcommand{\sumFqn}[1]{\sum_{\vec{#1} \in \F_q^n}}
\newcommand{\sumErrors}{\sumFqn{e}}
\newcommand{\sumCode}{\sum_{\cv \in \CC}}
\newcommand{\sumDual}{\sum_{\cv^\perp \in \CC^\perp}}

\newcommand{\Rbra}[1]{\left( #1 \right)} \newcommand{\Sbra}[1]{\left[ #1 \right]} \newcommand{\Cbra}[1]{\left\{ #1 \right\}} \newcommand{\Abra}[1]{\left\langle #1 \right\rangle} 

\newcommand{\floor}[1]{\left\lfloor #1 \right\rfloor}
\newcommand{\ceil}[1]{\left\lceil #1 \right\rceil}
\newcommand{\ints}[1]{\llbracket #1 \rrbracket}

\newcommand{\gauss}[2]{\genfrac{[}{]}{0pt}{}{#1}{#2}}

\renewcommand{\norm}[1]{\left\| #1 \right\|}
\DeclareMathOperator*{\Vol}{Vol}

\newcommand{\poly}{\mathsf{poly}}

\newcommand{\omeasy}{\omega_{\textup{easy}}}

\newcommand{\dgv}{d_{\textup{GV}}}
\newcommand{\drgv}{\delta_{\textup{GV}}}
\newcommand{\dmin}{d_{\textup{min}}}
\newcommand{\dtr}{D_{\textup{tr}}}
\newcommand{\dstat}{D_{\textup{stat}}}

\usepackage{scalerel,stackengine}
\stackMath
\newcommand\reallywidehat[1]{\savestack{\tmpbox}{\stretchto{\scaleto{\scalerel*[\widthof{\ensuremath{#1}}]{\kern.1pt\mathchar"0362\kern.1pt}{\rule{0ex}{\textheight}}}{\textheight}}{2.4ex}}\stackon[-6.9pt]{#1}{\tmpbox}}

\newcommand{\QFT}[1]{\reallywidehat{\ket{#1}}}
\newcommand{\mQFT}[1]{\widehat{#1}}
\newcommand{\fperp}{\mQFT{f}}
\newcommand{\ftruncperp}{\mQFT{f^{\text{Trunc}}}}

\newcommand{\piberp}{\pi^{\text{SC}}}
\newcommand{\piber}{\ket{\piberp}}
\newcommand{\piberQFT}{\mQFT{\piber}}
\newcommand{\piQFT}{\mQFT{\piber}}
\newcommand{\piQFTp}{\widetilde{\piber}}

\newcommand{\pitruncp}{\pi^{\text{Trunc}}}
\newcommand{\pitrunc}{\ket{\pitruncp}}
\newcommand{\pitruncQFT}{\mQFT{\pitrunc}}
\newcommand{\pitruncQFTp}{\widetilde{\pitrunc}}

\newcommand{\pitrunce}{\pitruncp_{\ev}}
\newcommand{\pitrunceQFT}{\mQFT{\pitrunce}}

\newcommand{\piunif}{\ket{\pi^{\textup{unif}}}}

\newcommand{\psiAp}{\psi_{\Ac}}
\newcommand{\psiA}{\ket{\psiAp}}
\newcommand{\psiAQFT}{\mQFT{\psiA}}

\newcommand{\psiidealp}{\psi_{\text{ideal}}}
\newcommand{\psiideal}{\ket{\psiidealp}}
\newcommand{\psiidealQFT}{\mQFT{\psiideal}}

\newcommand{\piun}{\braket{\pi}{\unv}^2}

\newcommand{\OO}[1]{O\Rbra{#1}}

\newcommand{\Th}[1]{\Theta\Rbra{#1}}

\newcommand{\Om}[1]{\Omega\Rbra{#1}}

\makeatletter
\newcommand*{\transp}{{\mathpalette\@transpose{}}}
\newcommand*{\@transpose}[2]{\raisebox{\depth}{$\m@th#1\intercal$}}
\makeatother
\newcommand{\transpose}[1]{{#1}^{\transp}}

\newcommand{\decP}{\textup{$\mathsf{DP}$}}
\newcommand{\LPN}{\textup{$\mathsf{LPN}$}}
\newcommand{\LWE}{\textup{$\mathsf{LWE}$}}
\newcommand{\SCP}{\textup{$\mathsf{SCP}$}}
\newcommand{\SIS}{\textup{$\mathsf{SIS}$}}

\newtheorem{theorem}{Theorem}
\newtheorem*{theorem-non}{Theorem}
\newtheorem{problem}{Problem}

\newtheorem{remark}{Remark}

\newtheorem{proposition}{Proposition}
\newtheorem{lemma}{Lemma}

\newtheorem{fact}{Fact}

 \makeatletter
\newcommand\bibalias[2]{\@namedef{bibali@#1}{#2}}

\newtoks\biba@toks
\newcommand\acite[2][]{\biba@toks{\cite#1}\def\biba@comma{}\def\biba@all{}\@for\biba@one:=#2\do{\@ifundefined{bibali@\biba@one}{\edef\biba@all{\biba@all\biba@comma\biba@one}}{\PackageInfo{bibalias}{Replacing citation `\biba@one' with `\@nameuse{bibali@\biba@one}'
      }\edef\biba@all{\biba@all\biba@comma\@nameuse{bibali@\biba@one}}}\def\biba@comma{,}}\edef\biba@tmp{\the\biba@toks{\biba@all}}\biba@tmp
}
\makeatother

\bibalias{Ouroboros-R}{AABBBDGHZ17}
\bibalias{RQC}{AABBBDGZ17}
\bibalias{RQC2}{AABBBDGZCH19}
\bibalias{ROLLO}{ABDGHRTZABBBO19}
\bibalias{LAKE}{ABDGHRTZ17}
\bibalias{LOCKER}{ABDGHRTZ17a}
\bibalias{RankSign}{AGHRZ17}
 
\author{Thomas Debris--Alazard$^{1,2}$} \email{thomas.debris@inria.fr}  
\author{Maxime Remaud$^{1,3}$} \email{maxime.remaud@atos.net}
\author{Jean-Pierre Tillich$^{1}$} \email{jean-pierre.tillich@inria.fr}
\address{$^{1}$ Inria}
\address{$^{2}$ Laboratoire LIX, \'Ecole Polytechnique, Institut Polytechnique de Paris, 1 rue Honor\'e d'Estienne d'Orves, 91120 Palaiseau Cedex}
\address{$^{3}$ Atos Quantum Lab}
\thanks{The work of TDA and JPT was funded by the French Agence Nationale de la Recherche through ANR JCJC COLA (ANR-21-CE39-0011) for TDA and ANR CBCRYPT (ANR-17-CE39-0007) for JPT.} 

\title{Quantum Reduction of Finding Short Code Vectors to the Decoding Problem}

\begin{document}
	\maketitle
	\begin{abstract}
		We give a quantum reduction from finding short codewords in a random linear code to decoding for the Hamming metric. This is the first time such a reduction (classical or quantum) has been obtained. Our reduction adapts to linear codes Stehl\'{e}-Steinfield-Tanaka-Xagawa’ re-interpretation of Regev's quantum reduction from finding short lattice vectors to solving the Closest Vector Problem. The Hamming metric is a much coarser metric than the Euclidean metric and this adaptation has needed several new ingredients to make it work. For instance, in order to have a meaningful reduction it is necessary in the Hamming metric to choose a very large decoding radius and this needs in many cases to go beyond the radius where decoding is always unique. Another crucial step for the analysis of the reduction is the choice of the errors that are being fed to the decoding algorithm. For lattices, errors are usually sampled according to a Gaussian distribution. However, it turns out that the Bernoulli distribution (the analogue for codes of the Gaussian) is too much spread out and cannot be used, as such, for the reduction with codes. This problem was solved by using instead a truncated Bernoulli distribution.
	\end{abstract}

	\section{Introduction} 

\par{\bf Code-based Cryptography.} 
Many cryptosystems such as public-key encryption schemes \cite{M78,A11,MTSB12}, authentication protocols \cite{S93} or pseudorandom generators \cite{FS96} are built relying on the hardness of finding the closest codeword, a task called {\em decoding}. In the case of a random linear code, which is the standard case, this problem can be expressed as follows:

\begin{problem}[$\decP(q,n,k,t)$] The decoding problem with  parameters  $q,n,k,t\in\N$ is defined as:
	\begin{itemize}
		\item Given: $(\Gm,\uv \Gm + \ev)$
		where $\Gm \in \Fq^{k \times n}$ and $\uv \in \Fq^k$ are sampled uniformly at random over their domain and $\ev \in \Fq^n$ over the words of weight $t$,
		\item Find: $\ev$.
	\end{itemize}	
\end{problem} 

This problem really corresponds to decoding the $k$-dimensional vector space $\CC$ ($\ie$, the code) generated by the rows of $\Gm$:
\begin{equation}\label{eq:def_code}
	\CC \eqdef \Cbra{\uv \Gm \colon \uv \in \Fq^k}.
\end{equation}

We are given the noisy codeword $\cv + \ev$ where $\cv$ belongs to $\CC$ and we are asked to find the error $\ev$ (or what amounts to the same, the original codeword $\cv$). This problem for random codes has been studied for a long time and despite many efforts on this issue, the best algorithms are exponential in the codelength $n$ in the regime where $t$ and $k$ are linear in $n$. 

Usually this decoding problem is considered, for the Hamming weight $|\vec{e}| \eqdef \sharp \{ i \in \llbracket 1,n \rrbracket, \quad e_{i}\neq 0 \}$,  in the regime where the code rate $R \eqdef \frac{k}{n}$ is fixed in $(0,1)$ and $q=2$, but there are also other interesting parameters for cryptographic applications. For instance, the Learning Parity with Noise problem ($\LPN$) corresponds to $\decP(2,n,k,t)$ where $n$ is the number of samples, $k$ the length of the secret while the error is sampled according to a Bernoulli distribution of fixed rate $t/n$. As the number of samples in $\LPN$ is unlimited, this problem really corresponds to decoding a code of rate arbitrarily close to $0$. 

While the security of many code-based cryptosystems relies on the hardness of the decoding problem, it can also be based on finding a ``short'' codeword (as in \cite{MTSB12} or in \cite{AHIKV17,BLVW19,YZWGL19} to build collision resistant hash functions), a problem which is stated as follows. 

\begin{problem}[$\SCP(q,n,k,w)$] The short codeword problem with parameters $q,n,k,w \in \N$ is defined as:
	\begin{itemize}
		\item Given:  $\Hm \in\Fq^{(n-k)\times n}$ which is sampled uniformly at random,		
		\item Find: $\cv \in \Fq^n$ such that $\Hm \transpose{\cv} = \zerov$ and the weight of  $\cv$ belongs to $\ints{1,w}$.
	\end{itemize}	
\end{problem} 

Here we are looking for a non-zero codeword $\cv$ of weight $ \leq w$ in the $k$-dimensional code $\CC$ defined by the so-called parity-check matrix $\Hm$, namely:
\begin{equation*}
	\CC \eqdef \Cbra{\cv \in \Fq^n \colon \Hm \transpose{\cv} = \zerov}.
\end{equation*}

Decoding and looking for short codewords are problems that have been conjectured to be extremely close. They have been studied for a long time \cite{P62,S88,D89,MMT11,BJMM12,MO15,BM18,CDMT22}, and for instance in the regime of parameters where the code rate $R$ is fixed in $(0,1)$, the best algorithms for solving them are the same (namely Information Set Decoding). A reduction from decoding to the problem of finding short codewords is known but in an $\LPN$ context \cite{AHIKV17,BLVW19,YZWGL19,DR22}. However, even in an $\LPN$ context, no reduction is known in the other direction. These problems can be viewed in some sense as a code version of the $\LWE$ and $\SIS$ problems respectively in lattice-based cryptography \cite{R09}. Our contribution in this article is precisely to give the code-based version of this reduction, namely a quantum reduction  from finding short codewords to decoding. This problem was open for quite some time. To simplify the statements, we will state it in the regime of parameters where the rate $R$ is fixed in $(0,1)$, but actually it also works in the $\LPN$ setting (but needs to be adapted in several places where we use exponential bounds in $n$).

There is a fundamental difficulty of reducing the research of low weight codewords to decoding a linear code which is due to the fact that the nature of these two problems is very different. Decoding concentrates on a region of parameters where there is typically just one solution, whereas finding low weight codewords concentrates on a region of parameters where there are solutions (and typically an exponential number of solutions). This makes these problems inherently very different. This was also the case for the reduction of $\SIS$ to $\LWE$ and the fact that we can have a reduction from one to another by looking for quantum reductions instead of classical reductions was really a breakthrough at that time.

\par{\bf Parameter range for $\decP$ and $\SCP$.}
An important parameter for the reduction is the decoding distance parameter $t$. The largest value of $t$ for which the decoding problem is ensured to have a unique solution is equal to $\floor{\frac{\dmin-1}{2}}$ where $\dmin \eqdef \min \{d(\cv,\cv') \colon \cv \in \CC, \; \cv' \in \CC, \;\cv \neq \cv'\}$ is the minimum distance of $\CC$ (which depends of course on the metric $d(\cdot,\cdot)$ that is considered). Standard probabilistic arguments can be used to show that the  minimum distance of a random linear code (code $\CC$ obtained as in Equation \eqref{eq:def_code} by a generator matrix $\Gm$ chosen uniformly at random in $\Fq^{k \times n}$) is with very high probability equal, up to an additive constant, to the {\em Gilbert-Varshamov distance} $\dgv(n,k)$ (or simply $\dgv$ if there is no ambiguity). It is defined (for all translation invariant metrics) for a code of dimension $k$ and length $n$, as the largest integer $t$ for which
\begin{equation}
	q^k  B_t \leq q^n
\end{equation}
where $B_t$ is the size\footnote{Note that because of the translation invariance, this size does not depend on the center.} of a ball of radius $t$ . It is generally convenient to consider the normalized Gilbert-Varshamov distance defined as $\drgv(n,k) \eqdef \frac{\dgv(n,k)}{n}$. For the Hamming metric, we have
\begin{equation*}
	\drgv(n,k) = h_q^{-1} \Rbra{1-\frac{k}{n}} + \OO{\frac{1}{n}} \quad \text{where}
\end{equation*}
\begin{equation*}
	h_q(x) \eqdef -x \log_q \Rbra{\frac{x}{q-1}} - (1-x) \log_q(1-x) \text{ and $h_q^{-1}$ is its inverse, ranging over $\Sbra{0,\frac{q-1}{q}}$.}
\end{equation*}

This Gilbert-Varshamov distance also happens to quantify the region where we {\em typically} have unique decoding. More precisely, it turns out that for many metrics of interest, including the Hamming metric, the same probabilistic arguments also show that the solution to the decoding problem is unique with probability $1 - 2^{-\Omega(n)}$ and when $n$ goes to infinity for fixed positive $\varepsilon$, $t \leq (1-\varepsilon) \dgv(n,k)$.

The best algorithms for solving the decoding problem have exponential complexity in $n$ as soon as $t$ is linear in $n$ and the code rate $R$ is bounded away from $0$ and $1$. This is in contrast to  the short codeword problem which becomes easy when the weight $w$ is above a certain range. The reason is that it is easy to produce codewords of small weight by using the fact that the code is a vector space of dimension $k$. Thus we can just produce codewords with $k-1$ entries equal to $0$ by solving a linear system which gives good candidates for having a small weight. It is straightforward that this strategy produces in polynomial time, for instance with the Hamming metric, codewords of weight $\approx \omeasy(n,k) n$ where
\begin{equation}\label{eq:weasy}
	\omeasy(n,k) \eqdef \frac{q-1}{q} \Rbra{1-\frac{k}{n}}
\end{equation}

Obtaining larger weights is also readily obtained by choosing only part of the $k-1$ entries to be equal to $0$. It should be noted that below $\omeasy(n,k)$ the best known algorithms for solving this problem have all exponential complexity for a fixed rate $R$ and a fixed ratio $\omega = \frac{w}{n}$.

\par{\bf Regev's quantum reduction strategy adapted to coding theory.} 
In \cite{R05} (see also the extended version \cite{R09}) Regev showed how to transform a random oracle solving the decoding problem in a lattice into a quantum algorithm outputting a rather small vector in the dual lattice. Our aim is to show here that the natural translation of this approach in coding theory gives an algorithm that outputs a rather small vector in the dual code. Roughly speaking Regev's approach relies on a fundamental result about the Fourier transform. 

\begin{proposition}\label{prop:constant_on_coset_and_dual}
	Consider an Abelian group $G$ and a function $f : G \mapsto \C$ that is constant on the cosets of a subgroup $H$ of $G$. Then the Fourier transform $\reallywidehat{f}$ is constant on the dual subgroup $H^\perp$.
\end{proposition}

This innocent looking fact, together with the fact that the quantum Fourier transform (QFT for short) can be performed in polylog time when the group $G$ is Abelian, is arguably the key to several remarkable quantum algorithms solving in polynomial time the period finding in a vectorial Boolean function \cite{S94d}, the factoring problem \cite{S94a} or the discrete logarithm problem \cite{S94a}. All of these problems can be rephrased in terms of the hidden Abelian subgroup problem, where one is given such a function $f$ that is constant (and distinct) on the cosets of an unknown subgroup $H$ and one is asked to recover $H$. This is achieved by 
\begin{itemize} 
	\item[$(i)$] creating the uniform superposition $\frac{1}{\sqrt{\abs{G}}} \sum_{x \in G} \ket{x} \ket{f(x)}$, 
	\item[$(ii)$] measuring the second register and discarding it, yielding a quantum state of the form $\frac{1}{\sqrt{\abs{H}}} \sum_{h \in H} \ket{x+h}$, 
	\item[$(iii)$] applying the QFT to it yielding a superposition of elements in the dual subgroup $H^\perp$ (and therefore gaining information on $H$ in this way). 
\end{itemize} 

Proposition \ref{prop:constant_on_coset_and_dual} is used in a similar way in Regev's reduction. Translating Regev's reduction in coding theory would use this framework by considering that the linear code $\CC$ we want to decode plays the role of the aforementioned $H$. From now on we will assume that this code is of dimension $k$ and length $n$ over $\Fq$. The algorithm would basically look as follows for reducing the search of small codewords in the dual code $\CC^\perp = \Cbra{\cv^\perp \in \Fq^n \colon \cv \cdot \cv^\perp = 0, \; \forall \cv \in \CC}$ (where $\xv \cdot \yv = \sum_{i=1}^n x_i y_i$ is the standard inner product in $\Fq^n$) to decoding errors of weight $t$ in $\CC$.

\begin{itemize}
	\item[Step 1.] Use a quantized version of the decoding algorithm to prepare the state 
		$$\frac{1}{\sqrt{Z}}\sum_{\cv \in \CC, \ev \in \Fq^n} \pi_{\ev} \ket{\cv + \ev}$$
		where $Z$ is a normalizing constant and $(|\pi_{\ev}|^2)_{\ev}$ is a probability distribution on errors that concentrate around the weight $t$ we are able to decode. This is done 
		\begin{itemize} 
			\item[$(i)$] by preparing first a superposition of codewords and errors, 
				$$\frac{1}{\sqrt{Z}} \sumCode \sumErrors \pi_{\ev} \ket{\cv}\ket{\ev},$$ 
			\item[$(ii)$] then adding the second register to the first one to get the entangled state 
				$$\frac{1}{\sqrt{Z}} \sumCode \sumErrors \pi_{\ev} \ket{\cv + \ev}\ket{\ev}$$
			\item[$(iii)$] and finally disentangling it thanks to a quantized version of the decoding algorithm, which from $\cv + \ev$ recovers $\ev$ and subtracts it from the second register to get the state
				$$\frac{1}{\sqrt{Z}} \sumCode \sumErrors \pi_{\ev} \ket{\cv+\ev}\ket{\zerov}.$$
		\end{itemize} 
	\item[Step 2.] Apply the QFT on $\Fq^n$ to obtain a superposition of elements $\cv^\perp$ in the dual code 
		$$\sumDual \alpha_{\cv^\perp}\ket{\cv^\perp}.$$
	\item[Step 3.] Measure the register to output $\cv^\perp$ of rather small norm in $\CC^\perp$.
\end{itemize}

The second step is a direct consequence of Proposition \ref{prop:constant_on_coset_and_dual}. The last one raises the issue of whether or not the QFT concentrates the weight of the vector output by this algorithm on weights $t'$ for which finding a codeword in $\CC^\perp$ is not known to be easy, as it is the case for Regev's reduction on lattices equipped with the Euclidean metric.

\par{\bf On the difficulty of translating Regev's reduction to the Hamming metric.}
This thread of research has been pioneered by Yilei Chen \cite{C20a} and later on in \cite{CV20}, where basically the following approach was taken. The natural analog in the Hamming metric case of the Gaussian noise model used in Regev's reduction \cite{R09} is the $q$-ary symmetric channel. Its associated quantum state is given by 
\begin{equation*}
	\piber \eqdef \sumErrors (1-\tau)^{\frac{n - \abs{\ev}}{2}} \Rbra{\frac{\tau}{q-1}}^{\frac{\abs{\ev}}{2}} \ket{\ev} = \Rbra{\sqrt{1-\tau} \ket{0} + \sumFq{\alpha} \sqrt{\frac{\tau}{q-1}} \ket{\alpha}}^{\otimes n}
\end{equation*}
where $\abs{\ev}$ stands for the Hamming weight of $\ev$, $n$ for the length of $\ev$ and $\tau$ is the crossover probability of the $q$-ary symmetric channel. Indeed, measuring such a state yields an error distributed like a $q$-ary symmetric channel of crossover probability $\tau$. In both cases (be it for the Gaussian noise or the $q$-ary symmetric channel), the Fourier transform yields a dual noise which is again Gaussian or $q$-ary symmetric respectively and the quantum state corresponding to the error is a product state which considerably simplifies the computation. In the case of the $q$-ary symmetric noisy chanel, applying the QFT on $\piber$ yields the quantum state
\begin{equation*}
	\sumErrors (1-\tau^\perp)^{\frac{n - \abs{\ev}}{2}} \Rbra{\frac{\tau^\perp}{q-1}}^{\frac{\abs{\ev}}{2}} \ket{\ev} = \Rbra{\sqrt{1-\tau^\perp} \ket{0} + \sumFq{\alpha} \sqrt{\frac{\tau^\perp}{q-1}} \ket{\alpha}}^{\otimes n} \eqdef \piberQFT
\end{equation*}
where (see Fact \ref{fa:Fourier}) 
$$
\tau^{\perp} \eqdef \frac{\Rbra{\sqrt{(q-1)(1-\tau)} - \sqrt{\tau}}^{2}}{q}.$$ 
This new quantum state represents a $q$-ary symmetric channel of parameter $\tau^\perp$. If we measure $\piber$ we get relative weights $\approx \tau$ whereas if we measure $\piberQFT$ we get relative weights around $\tau^\perp$. 

It would thus be tempting to conclude that the ``ideal'' version of the algorithm presented above will output dual codewords (in Step $3$) of relative Hamming weight $\approx \tau^{\perp} < \omeasy $, $\ie$, in the regime where there is a chance that it is difficult to produce such words. However, this natural approach runs into the following problem. The parameter $\tau$ of the Bernoulli noise has to be chosen  so that the typical error weight $\tau n$ is equal to or slightly below the weight $t$ we can decode. Such a $\tau$ is therefore at most the relative Gibert-Varshamov $\drgv$. However \cite{CV20} proved that in this case the most likely relative weight we measure at Step 3 is {\em typically zero} if $\tau^{\perp} < \omeasy$.  In other words, the straightforward application of Regev's approach to coding theory fails to give a useful reduction. 

We will give in Remark \ref{rem:Conditions} another explanation for the failure of this approach. It can be summarized by saying that the Bernoulli noise model is not concentrated enough on its typical weight $\tau n$.

\par{\bf Our approach.}
To tackle this issue, it would therefore be natural to choose the most concentrated noise model on the weight $t$, namely:
\begin{equation*}
	\piunif = \sum_{\ev \colon \abs{\ev} = t} \frac{1}{\sqrt{S_t}} \ket{\ev}
\end{equation*}
where $S_t$ is the cardinality of the sphere of radius $t$ in the Hamming metric, $\ie$, $S_t = (q-1)^t \binom{n}{t}$. Understanding of which weight $w$ is the outcome after measuring the state $\sumDual \alpha_{\cv^\perp} \ket{\cv^\perp}$ in Step 3 is more difficult in the constant weight error model than in the Bernoulli noise model. In particular, it involves properties of Krawtchouk polynomials. However, it can be shown that when $\omega \eqdef \frac{w}{n}$ lies in a whole interval $\Sbra{\tau^\perp, \tau_{+}^\perp}$ where $\tau_{+}^{\perp} \eqdef \frac{\Rbra{\sqrt{(q-1)(1-\tau)} + \sqrt{\tau}}^2}{q}$ 
, we have many points where the probability of measuring a word of weight $w$ is actually $\frac{1}{\poly(n)}$. The ``dual'' weight distribution is not really concentrated on a single value but spread over a large interval. This would provide a useful reduction when we use a decoding algorithm that succeeds on a non-negligible set of inputs. 

Unfortunately, to be relevant in a cryptographic context, we must consider the case where decoding succeeds only for a potentially very low probability $\varepsilon$, the aim being to turn our decoding algorithm into an algorithm that produces a low weight codeword from the dual with some probability $\poly(\varepsilon)$. This cannot be obtained with the uniform distribution on the sphere of radius $t$. Indeed, the ``ideal'' version of the algorithm we presented before (where we assume we always succeed with our decoding algorithm) describes a state we obtain in Step 2 that is not completely orthogonal (the scalar product is bounded from below by a quantity $\poly(\varepsilon)$) to the ``real'' state after applying this approximate decoding process and the QFT. If we were to measure this state directly (starting from the uniform noise model over the sphere of radius $t$), we would only be sure to measure a word of some relative weight lying in the interval $\Sbra{\tau^\perp, \tau_{+}^{\perp}}$ with probability $\poly(\varepsilon)$, since the ``ideal'' state concentrates its relative weight distribution in all this interval.
In this way, we cannot ensure the measurement of a dual codeword of smallest possible relative weight, namely $\tau^{\perp}$, it could be $1/2 \in \Sbra{\tau^\perp, \tau_{+}^{\perp}}$. We really need here a distribution that is rather sharply concentrated around the decoding radius of our decoding algorithm, but whose Fourier transform is also sharply concentrated around a certain weight. 

So far, two noise models have been considered for the reduction to work, each with an advantage and a drawback:
\begin{itemize}
	\item[$\bullet$] the $q$-ary symmetric noise $\piber$ is not concentrated enough on its typical weight $\tau n$ {\em but} its dual noise is sufficiently concentrated on $\tau^{\perp} n$,
	\item[$\bullet$] the uniform noise $\piunif$ on the sphere of radius $\tau n$ is sufficiently concentrated {\em but} its dual noise is spread out on the whole interval $\Rbra{\tau^{\perp}n, \tau_{+}^{\perp}n}$. 
\end{itemize}

Interestingly, the issues with these two distributions are opposite. Fortunately, it turns out that we have a natural noise model to get a best-of-both-worlds model: {\em truncating} the $q$-ary symmetric noisy channel. More precisely, consider the following noise model (for some small enough constant $\varepsilon >0$),
\begin{equation*}
	\pitrunc = \frac{1}{\sqrt{N}} \sum_{\substack{\ev \colon \\ \abs{\ev} \in \Sbra{(1-\varepsilon)t, (1+\varepsilon)t}}} (1-\tau)^{\frac{n - \abs{\ev}}{2}} \Rbra{\frac{\tau}{q-1}}^{\frac{\abs{\ev}}{2}} \ket{\ev}
\end{equation*}

where $N$ is a normalizing factor. This noise model solves our above issues with $\piber$ and $\piunif$ because it verifies our two constraints for the reduction to work:
\begin{itemize}
	\item[$(i)$] its weight distribution is sufficiently concentrated around $\tau n$,
	\item[$(ii)$] its dual noise after applying the Fourier transform is concentrated on the relative weight $\tau^\perp n$. 
\end{itemize}

Contrary to $(i)$, assertion $(ii)$ may seem unclear. It relies on the following equality as we will show in Lemma \ref{lemma:NegliTruncBer}
\begin{equation*}
	\norm{\pitrunc - \piber}  = 2^{-\Omega(n)}
\end{equation*}
where $\norm{\cdot}$ stands for the norm of the Hilbert space in which the quantum states are embedded. Therefore, applying the Fourier transform (which is an isometry for $\norm{\cdot}$) on $\pitrunc$ will yield a quantum state which is $2^{-\Omega(n)}$-close of $\piberQFT$.

With our approach and the truncated $q$-ary symmetric channel, we transform through the QFT a decoding algorithm correcting $\tau n$ errors into an algorithm outputting with non-negligible probability words of weight $\approx \tau^\perp n$ in the dual code. The distance $\tau^\perp$ is clearly a decreasing function of $\tau$ and the issue is now whether or not there exists a $\tau < \drgv(n,k)$ (this is the biggest value for which we can hope that decoding is successful with probability $1-o(1)$) such that $\tau^\perp < \omeasy(n,n-k)$ (here we want to find short codewords in the dual code $\CC^\perp$ which is of dimension $n-k$), since this would yield a useful reduction. It turns out that in many cases we have to choose $\tau > \drgv(n,k)/2$, meaning that we are not in the regime where decoding necessarily has at most one solution. This complicates the proof of the reduction somewhat since with a quantized version of the decoding algorithm, we will not be able to produce at Step 1 the state $\frac{1}{\sqrt{Z}} \sumCode \sumErrors \pi_{\ev} \ket{\cv + \ev}$ (since decoding fails for some $\ev$) but we will show that as long as $\tau <  \drgv(n,k)$, we will get a state close to it. This will be enough for our purpose. 

By putting all these ingredients together, we are able to prove the following result.

\begin{theorem-non}[informal]
	The short codeword problem $\SCP(q,n,n-k,w)$ reduces to the decoding problem $\decP(q,n,k,t)$ for $w = \tau^\perp n + O(1)$ where 
	\begin{equation*}
		\tau \eqdef \frac{t}{n} \qquad \text{ and } \qquad \tau^\perp \eqdef \frac{\Rbra{\sqrt{(q-1)(1-\tau)} - \sqrt{\tau}}^2}{q}.
	\end{equation*}
\end{theorem-non}

It will turn out that for $q=2$ (see Section \ref{sec:useful}) we can find for any rate $R = \frac{k}{n}$ in $(0,1)$ a $t < \dgv(n,k)$ for which the corresponding $w$ is below $\omeasy(n,n-k) n$ (the reduction is useful in this case). It also corresponds to parameter ranges that are relevant for certain cryptographic applications (see \cite{AFS05} whose security relies on the $\SCP$ problem in a parameter range which is covered by our reduction and \cite{S93} which is an identification scheme whose security actually relies on the decoding problem just below the Gilbert-Varshamov distance). Unfortunately, this is not true anymore when $q \geq 5$, where there is always a range for $R$ for which $w$ is above $\omeasy(n,n-k) n$, for any choice of $t < \dgv(n,k)$: the reduction becomes useless in this case. Roughly speaking, when $q$ grows, the Hamming metric gets coarser (we have only $n+1$ different values for the metric on $\Fq^n$, whereas the size of the ambient space gets bigger) and this results in the range of values of $R$ for which this reduction is useful becoming smaller.

\par{\bf Considering other metrics.}
The whole approach we have followed here (properly choosing the error distribution and going beyond the unique decoding radius for decoding if necessary) can of course be adapted to other metrics. 
It is easy for instance to apply it to the rank metric which is becoming increasingly popular in code-based cryptography, see for instance \cite{ABDGHRTZABBBO19,AABBBDGZCH19,BCGMM19,BGHM20}. This metric is even coarser than the Hamming metric: on $\Fq^{m \times n}$ there are only $1 + \min(m,n)$ different values for the rank weight (given a matrix, it is defined as its rank). In this case, as we will see, the reduction is always useless ($\ie$, reduces to weights which are always easy to produce for a random linear code).

\section{Notation and some useful facts}

\par{{\bf General Notation.}}
For $a$ and $b$ integers with $a \leq b$, we denote by $\ints{a,b}$ the set of integers $\Cbra{a,a+1,\dots,b}$. We extend this notation when $a$ and $b$ are not integers to the set of integers in $[a,b]$. Vectors are in {\em row notation} and they will be written with  bold letters (such as $\ev$). Uppercase bold letters are used to denote matrices (such as $\Hm$). Notation $\Sc_t$ is the sphere of radius $t$ around $0$ in $\Fq^n$ (for a metric $|\cdot|$ that will be clear from the context) and $S_t$ is its cardinality. In all this paper, $\poly(n)$ denotes a quantity which is an $\OO{n^a}$ for some constant $a$.

\par{{\bf Subspaces and Gaussian binomial coeffient.}} 
$F \le \Fq^n$ denotes that $F$ is a subspace of $\Fq^n$. When $\Em$ is a matrix, we denote its rank by $\abs{\Em}$. 

The following asymptotic expressions will be extremely useful. For the rank metric, $S_{\ell}$ is equal to the number of matrices in $\Fq^{m \times n}$ of rank $\ell$. From this interpretation we obtain
\begin{equation*}
	S_{\ell} = \prod_{j = 0}^{\ell-1}(q^{m}-q^{j}) \gauss{n}{\ell}_q \quad \mbox{with} \quad \gauss{n}{t}_q =
	\begin{cases}
		\displaystyle \prod_{i=0}^{t-1} \frac{q^n-q^i}{q^t-q^i} & \mbox{ if } t \le n, \\
		0 & \mbox{ otherwise.}
	\end{cases}
\end{equation*}

It follows by using the fact that $\prod_{i=1}^\infty (1-q^{-i})$ is some constant depending on $q$, that as $n \to +\infty$ 
\begin{equation}\label{eq:asymptRank1}
	\gauss{n}{\ell}_q = \Th{q^{\ell(n-\ell)}},
\end{equation}
\begin{equation}\label{eq:asymptRank2}
	S_{\ell} = \Th{q^{\ell(m + n -\ell)}}.
\end{equation}

From \eqref{eq:asymptRank2} we deduce that 
\begin{equation}\label{eq:ratio}
	\frac{S_{u+1}}{S_u} = \Th{\frac{q^{(u +1)(m +n -u -1)}}{q^{u(m +n -u)}}} = \Th{q^{m + n - 2u -1}}.
\end{equation}

\par{{\bf Quantum Fourier Transform (QFT).}}
The Quantum Fourier Transform $\QFT{\psi}$ of a state $\ket{\psi} \eqdef \sumFqn{x} \alpha_{\xv} \ket{\xv}$ is defined by using the characters $\chi_{\yv}$ of the additive group $\Fq^n$ (there are as many characters as there are elements in $\Fq^n$ and we assume that the characteristic of $\Fq$ is the prime $p$ with $q = p^s$) as
\begin{equation*}
	\QFT{\psi} \eqdef \frac{1}{\sqrt{q^n}} \sumFqn{y} \widehat{\alpha_{\yv}} \ket{\yv} \qquad \text{with} \qquad \widehat{\alpha_{\yv}} \eqdef \sumFqn{x} \alpha_{\xv} \chi_{\yv}(\xv)
\end{equation*}
where
\begin{eqnarray*}
	\chi_{\yv}(\xv) & \eqdef & e^{\frac{2i \pi \Tr(\xv \cdot \yv)}{p}}, \quad \text{with} \\
	\xv \cdot \yv & \eqdef & \sum_{i=1}^n x_i y_i \quad \text{with $\xv=(x_i)_{i=1}^n$ and $\yv=(y_i)_{i=1}^n$} \quad \text{and} \\
	\Tr(a) & \eqdef & a + a^p + a^{p^2} + \dots + a^{p^{s-1}}.
\end{eqnarray*}

A useful computation is the QFT of a state representing the $q$-ary symmetric channel of crossover probability $\tau$:
\begin{fact}\label{fa:Fourier}
	Let $\piber \eqdef \Rbra{\sqrt{1-\tau} \ket{0} + \sumFq{\alpha} \sqrt{\frac{\tau}{q-1}} \ket{\alpha}}^{\otimes n}$, where $\tau \in \Sbra{0,\frac{q-1}{q}}$, then 
	\begin{eqnarray}
		\piberQFT & = & \Rbra{\sqrt{1-\tau^\perp} \ket{0} + \sumFq{\alpha} \sqrt{\frac{\tau^\perp}{q-1}} \ket{\alpha}}^{\otimes n}, \quad \text{where} \\
		\tau^\perp & \eqdef & \frac{\Rbra{\sqrt{(q-1)(1-\tau)} - \sqrt{\tau}}^{2}}{q}. \label{def:tauPerp}
	\end{eqnarray}
\end{fact}
\begin{proof}
	Let $\ket{\psi} \eqdef \sqrt{1-\tau} \ket{0} + \sumFq{\alpha} \sqrt{\frac{\tau}{q-1}} \ket{\alpha}$. Then, it is readily verified that
	\begin{equation*}
		\piberQFT = \Rbra{\mQFT{\ket{\psi}}}^{\otimes n} \qquad \text{ where } \qquad \mQFT{\ket{\psi}} = \frac{1}{\sqrt{q}} \Rbra{\beta_0 + \sumFq{y} \beta_y \ket{y}} \qquad \text{with}
	\end{equation*}
	\begin{equation*}
		\beta_0 = \sqrt{1-\tau} + (q-1) \sqrt{\frac{\tau}{q-1}} \qquad \text{ and } \qquad \beta_y = \sqrt{1-\tau} + \sqrt{\frac{\tau}{q-1}} \sumFq{x} \chi_y(x).
	\end{equation*}
It is easy to verify that for any $y \in \Fq^*$ we have for any $x_0 \in \Fq^\ast$:
	\begin{equation*}
 		\sum_{x \in \Fq} \chi_y(x) =  \sum_{x \in \Fq} \chi_y(x_0 \cdot x) = \chi_y(x_0) \sum_{x \in \Fq} \chi_y(x)
	\end{equation*}
	Since there exists $x_0 \in \Fq^\ast$ such that $\chi_y(x_0) \neq 1$, we deduce that $\sum_{x \in \Fq} \chi_y(x)=0$. Now since $\sum_{x \in \Fq} \chi_y(x)= 1 + \sumFq{x} \chi_y(x)$, we have $\sumFq{x} \chi_y(x) = -1$ for any $y$ in $\Fq^\ast$ and therefore
	\begin{eqnarray*}
		\frac{\beta_y}{\sqrt{q}} & = & \sqrt{\frac{1-\tau}{q}} -  \sqrt{\frac{\tau}{q(q-1)}}\\
		& = & \sqrt{\frac{(q-1)(1-\tau)}{q(q-1)}} -  \sqrt{\frac{\tau}{q(q-1)}}\\
		& = & \frac{\sqrt{\frac{(q-1)(1-\tau)}{q}} -  \sqrt{\frac{\tau}{q}}}{\sqrt{q-1}}\\
		& = & \frac{\sqrt{\tau^\perp}}{\sqrt{q-1}}.
	\end{eqnarray*}
	It is then clear that $\frac{\beta_0}{\sqrt{q}} = \sqrt{1 - \tau^\perp}$ from $\sum_{y \in \Fq} \abs{\beta_y}^2 = 1$, concluding the proof.
\end{proof}

The dual code $\CC^\perp$ of a linear code $\CC$ over $\Fq$ is easily seen to be defined equivalently from the characters or from the inner product $\xv \cdot \yv$ as follows:
\begin{align*}
	\CC^\perp & \eqdef \Cbra{\yv \in \Fq^n \colon \forall \cv \in \CC, \; \chi_{\yv}(\cv) = 1}\\
	& =  \; \Cbra{\yv \in \Fq^n \colon \forall \cv \in \CC, \; \yv \cdot \cv = 0}.
\end{align*}

 	\section{Quantum Reduction from Sampling Short Codewords to Decoding}

\subsection{A general result}

We assume here that we have a probabilistic algorithm $\Ac$ that solves (sometimes) the decoding problem at distance $t$. Its inputs are a generator matrix $\Gm \in \Fq^{k \times n}$ of a code $\CC \subseteq \Fq^n$ ($\ie$, $\CC = \Cbra{\uv \Gm \colon \uv \in \Fq^k}$) and a noisy codeword $\cv + \ev$ where $\cv$ belongs to $\CC$. We denote by $\rv \in \F_2^\ell$ the internal coins of $\Ac$. It outputs with a certain probability $\varepsilon$, the ``right'' $\ev$ when being fed with $\cv + \ev$ where $\cv$ and $\ev$ are uniformly chosen at random in $\CC$ and among the errors of weight $t$ respectively:
\begin{equation}\label{eq:epsilon} 
	\varepsilon \eqdef \Prob_{\Gm, \cv, \ev, \rv} \Rbra{\Ac(\Gm, \cv+\ev, \rv) = \ev}.
\end{equation}

The quantum reduction starts by building the initial superposition
\begin{equation*}
	\frac{1}{\sqrt{2^\ell q^k}} \sumErrors \sumCode \sum_{\rv \in \F_2^\ell} \pi_{\ev} \ket{\ev} \ket{\cv} \ket{\rv}
\end{equation*}
where $\ket{\pi} \eqdef \sumErrors \pi_{\ev} \ket{\ev}$ is some quantum superposition of errors. In practice, we will voluntarily omit to write the internal coins of $\Ac$ for the sake of readability and will therefore start from
\begin{equation*}
	\frac{1}{\sqrt{q^k}} \sumErrors \sumCode \pi_{\ev} \ket{\ev} \ket{\cv}
\end{equation*}
For the same reason, we will write $\varepsilon = \Prob_{\Gm, \cv, \ev} \Rbra{\Ac(\Gm, \cv+\ev) = \ev}$.

The quantum algorithm that gives the reduction can then be described as follows.

\begin{center}{\bf Algorithm of the quantum reduction.}\\
\noindent\fbox{\parbox{\textwidth}{
\begin{align}
	&\text{Initial state preparation} & = & \quad \frac{1}{\sqrt{q^k}} \sumErrors \sumCode \pi_{\ev} \ket{\ev} \ket{\cv} \nonumber \\
	&\text{adding $\ev$ to $\cv$:} & \mapsto & \quad \frac{1}{\sqrt{q^k}} \sumErrors \sumCode \pi_{\ev} \ket{\ev} \ket{\cv +\ev} \nonumber \\
	&\text{applying $\Ac$:} & \stackrel{\Ac}{\mapsto} & \quad \frac{1}{\sqrt{q^k}} \sumErrors \sumCode \pi_{\ev} \ket{\ev - \Ac(\Gm,\cv+\ev)} \ket{\cv+\ev} \eqdef \psiA \label{eq:A} \\
	&\text{QFT on the $2$nd register:} & \mapsto & \quad \psiAQFT \label{eq:QFT} \\
	&\text{measuring the whole state:} & \mapsto & \quad \ket{\ev} \ket{\cv^\perp}  \label{eq:measure}
\end{align}}}
\end{center}

We will now give a general theorem about an algorithm of this kind and will show that it produces a codeword of the dual code $\CC^\perp$ of some weight $u$ with probability $\poly(\varepsilon)$ when certain conditions are met.

\begin{restatable}{theorem}{thmain} \label{th:main}
	Assume that $\ket{\pi}$ is radial and non-negative, $\ie$, $\pi_\ev = f(\abs{\ev})$ for some non-negative function $f$. Assume that $\widehat{\ket{\pi}} = \sumErrors \mQFT{\prob} \ket{\ev}$ is radial too\footnote{In other words, we assume that the Fourier transform is radially preserving. This property depends on the characters chosen to define the Fourier transform and the metric. Recall that a radial function is a function which is constant on spheres centered around $0$. This property clearly holds for functions $f \colon \Fq^n \rightarrow \C$ with the characters chosen here and the Hamming metric. We give this more general statement in order to apply it in other cases of interest, for instance the rank metric.} and let $\fperp(w) = \mQFT{\prob}$ for any element $\ev$ of $\Fq^n$ of weight $w$. Furthermore, assume that there exists an interval $\Wc \subseteq \ints{0,n}$ such that:
	\begin{align}
		& \text{(Concentration of $\pi$)} & \frac{\piun}{q^{n-k}} & = 2^{-\Omega(n)}, \tag{$C1$}\label{eq:concentrationPrimal}\\
		& \text{(Exponentially many dual codewords of weight $u \in \Wc$)} & \sum_{u \in \Wc} \frac{q^k}{S_u} & = 2^{-\Omega(n)}, \tag{$C2$}\label{eq:aboveGV}\\
		& \text{(Concentration of the dual distribution $\widehat{\pi}$ on $\Wc$)} & \sum_{u \in \Wc} S_{u} \abs{\fperp(u)}^{2} & = 1 - 2^{-\Omega(n)} \tag{$C3$}\label{eq:concentrationDual}
	\end{align}
	with $\ket{\unv}$ being the (unnormalized) superposition of errors : $\ket{\unv} \eqdef \sumErrors \ket{\ev}$.

	Suppose that there exists an algorithm $\Ac$ solving the decoding problem $\decP(q,n,k,t)$ with success probability $\varepsilon$. Then, there exists a quantum algorithm which takes as input a generator matrix $\Gm \in \Fq^{k \times n}$ of $\CC$ and outputs  a codeword of weight $u \in \Wc$ in $\CC^\perp$ with probability greater than $\frac{p_t^2 \varepsilon^3}{16} - O(p_t^4 \varepsilon^5)  - 2^{- \Omega(n)} - \OO{q^{-\min(k,n-k)}}$ where $p_t \eqdef \sum_{\ev \colon \abs{\ev} = t} \abs{\pi_{\ev}}^2 = S_t f(t)^2$.
\end{restatable}

\begin{remark}\label{rem:Conditions}\leavevmode
	\begin{itemize}
		\item We will use this theorem for Hamming and rank metrics, but it can be applied to {\em any} metric for which the Fourier transform is radially preserving.
		\item When $\pi_{\ev}$ is non-negative, Condition \eqref{eq:concentrationPrimal} basically requires the probability distribution on $\Fq^n$, $\mu \eqdef (\pi_\ev^2)_{\ev \in \Fq^n}$, to be sufficiently concentrated. This quantity can be expressed as $q^k (1- H^2(\mu,U))^2$, where $U$ stands for the uniform distribution over $\Fq^n$ and $H(\pv,\qv) \eqdef \sqrt{1 - \sum_{i} \sqrt{p_i q_i}}$ is the Hellinger distribution between two probability distributions $\pv$ and $\qv$ defined over a same probability space.
		\begin{enumerate}
			\item It is clearly maximal for the uniform probability distribution over $\Fq^n$, and
			\item on the other hand, when considering the Hamming metric and $\ket{\pi} = \piber$, we have 
				\begin{equation*}
					\frac{\braket{\piberp}{\unv}^{2}}{q^{n-k}} \; = \; \frac{q^k}{q^n} \abs{\sumFqn{y} \chi_{\yv}(\zerov) \pi_{\yv}}^2 \; = \; q^{k} \abs{\fperp(0)}^2 \; = \; q^{k}(1-\tau^{\perp})^n.
				\end{equation*}
				It can be verified that there is no way to choose $\tau$ such that at the same time:
				\begin{itemize} 
					\item[$(i)$] $\tau n \leq \dgv(n,k)$ (otherwise there is no hope to decode correctly most of the time),
					\item[$(ii)$] $\tau^\perp \leq \omeasy$ (otherwise finding codewords in $\CC^\perp$ of weight $\tau^\perp n$ is easy),
					\item[$(iii)$] $q^k (1-\tau^\perp)^n = o(1)$.
				\end{itemize}
				The quantity $\frac{\braket{\piberp}{\unv}^{2}}{q^{n-k}}$ is just too big, or in other words, the distribution of $\mu$ is too much spread out and not concentrated enough on its typical weight.
		\end{enumerate}
		\item Condition \eqref{eq:aboveGV} expresses that $u$ lies in a subset of values for which a random $[n,n-k]$-code has an exponential expected number of codewords. Indeed, the expected number of codewords of weight $u$ is equal to $\frac{S_u}{q^k}$.
		\item Finally, Condition \eqref{eq:concentrationDual} expresses that the dual probability distribution $\widehat{\mu} \eqdef (\abs{\widehat{\pi}_\ev}^2)_{\ev \in \Fq^n}$ is almost completely supported on $\Wc$ up to an exponentially small vanishing term.
	\end{itemize}
\end{remark}

\subsection{Outline of the proof of Theorem \ref{th:main}}

Let us first give a general outline of the proof before detailing each step.
\begin{itemize}
	\item[{\bf Step 1.}] We prove that after applying $\Ac$ in the reduction, $\psiA$ is close enough to the ``disentangled'' state
		\begin{equation}\label{eq:psiIdeal}  
			\psiideal \eqdef \frac{1}{\sqrt{Z}} \sumErrors \sumCode \pi_{\ev} \ket{\zerov_n} \ket{\cv+\ev}
		\end{equation} 
		where $Z$ is a normalizing constant.
	\item[{\bf Step 2.}] We then analyze the effect of the QFT on the ``ideal state'' $\psiideal$ and a subsequent measurement of it. We namely prove that measuring it produces a codeword $\cv^\perp \in \CC^\perp$ of weight $u$ with probability $\abs{\fperp(u)}^2$, up to a normalizing factor.
	\item[{\bf Step 3.}] Then we prove that the number of codewords of weight $u$ in $\CC^\perp$ is typically very close to $\frac{S_u}{q^k}$. With Step 2 and the assumptions of Theorem \ref{th:main}, we infer that the probability of observing a dual codeword of weight in the set $\Wc$ after measuring $\psiidealQFT$ is exponentially close to $1$.
	\item[{\bf Step 4.}] We upper-bound the statistical distance between the probability distribution of the states after measuring $\psiAQFT$ and $\psiidealQFT$ respectively by using Step 1 and the properties of the trace distance given in Fact \ref{fact:trace} below. 
\end{itemize}

Let us give more details about these steps.

\par{\bf Step 1.}
For this purpose, we use the trace distance between quantum states (as in  \cite{SSTX09} where this has been used in the lattice setting). It is defined 
as follows :
\begin{equation}\label{eq:def_dist_tr}
	\dtr(\ket{\phi},\ket{\psi}) \eqdef \sqrt{1 -\abs{\braket{\phi}{\psi}}^2}.
\end{equation}
This distance meets the following properties that will prove useful in our context:
\begin{fact}\label{fact:trace}\leavevmode
	\begin{enumerate}[label=\upshape(\Roman*),ref=\thefact (\Roman*)]
		\item\label{prop:trace1} It can never increase after a quantum evolution \cite[\S9,Th. 9.1]{NC16};
		\item\label{prop:trace2} The pair of probability distributions $(p_m, q_m)$ of the measurement outcome $m$ of any quantum measurement performed on the pair of states $(\ket{\phi}, \ket{\psi})$ satisfies \cite[\S9,Th. 9.2]{NC16}
			\begin{equation}\label{eq:dstat_dtr}
				\dstat(p_m,q_m) \leq \dtr(\ket{\phi}, \ket{\psi})
			\end{equation}
			where $\dstat$ is the statistical distance (also called the total variation distance) between two probability distributions. It is defined by:
			\begin{equation*}
				\dstat(p,q) \eqdef \frac{1}{2} \sum_{x \in \Xc} \abs{p(x)-q(x)}
			\end{equation*}
			where $p$ and $q$ are two discrete probability distributions on $\Xc$.
	\end{enumerate}
\end{fact}

With this notion we can prove that
\begin{proposition}\label{prop:step1}
	With probability greater than $1 - \frac{\piun}{q^{n-k}} - \OO{q^{-\min(k, n-k)}}$ over the choices of $\Gm$ we have:
	\begin{equation*}
		\dtr(\psiA,\psiideal) \leq \sqrt{1 - \frac{p_t^2}{2} \varepsilon_{\Gm}^2 },
	\end{equation*}
	where $\varepsilon_{\Gm}$ is the probability that $\Ac$ returns the right error $\ev$ when the input matrix is $\Gm$, $\ie$,
	\begin{equation}\label{eq:epsG}
		\varepsilon_{\Gm} \eqdef \Prob_{\cv,\ev} \Rbra{\Ac(\Gm, \cv+\ev) = \ev}.
	\end{equation}
\end{proposition}

The proof of this result follows immediately from three lemmas (whose proof is in Appendix \ref{app:mainStep1}). The first one bounds the trace distance in terms of $\varepsilon_{\Gm}$ and $Z$, and the second one gives a tight upper-bound on the expected value of $Z$ for a related probabilistic model. The latter is used to derive the third one which is fundamental and states that it is very unlikely for $Z$ to be much greater than the ``natural'' constant $q^k$:
\begin{restatable}{lemma}{lemdistalgoideal}\label{lem:dist_algo_ideal}
	We have:
	\begin{equation*}
		\dtr(\psiA, \psiideal) \leq \sqrt{1 - \frac{q^k p_t^2}{Z} \varepsilon_{\Gm}^2 }.
	\end{equation*}
\end{restatable}

\begin{restatable}{lemma}{lemEZ}\label{lem:EZ}
	Assume that $\CC$ is chosen by uniformly drawing at random a parity-check matrix $\Hm$ for it. We have:
	\begin{equation}\label{eq:ub_on_EZ}
		\esp(Z) \leq q^k \Rbra{1+ \frac{\piun}{q^{n-k}}}.
	\end{equation}
\end{restatable}

\begin{restatable}{lemma}{lemZ}\label{lem:Z}
	Let $\eta > 0$. We have:
	\begin{equation*}
		\Prob_{\Gm}(Z > q^k (1+\eta)) \leq \frac{1}{\eta} \; \frac{\piun}{q^{n-k}} + \OO{q^{-\min(k,n-k)}}.
	\end{equation*}
\end{restatable}

Proposition \ref{prop:step1} immediately follows by using $\eta=1$ in Lemma \ref{lem:Z} and plugging this bound on $Z$ in Lemma \ref{lem:dist_algo_ideal}. The quantity $q^k$ is the natural value for $Z$ since it is what we can expect when all the $\cv+\ev$ terms (taking all $\cv$ in $\CC$ and all typical $\ev$) are different. The constant $Z$ increases precisely when there are many collisions for the $\cv+\ev$ terms. However, in this case, we do not expect to be able to solve the decoding problem anymore.

\par{\bf Step 2.}
More precisely, we prove that
\begin{restatable}{lemma}{lemmeasure}\label{lem:measure}
	If the Fourier transform is radially preserving, meaning that it transforms a radial function into a radial function, then after measuring $\psiidealQFT$ we obtain a state $\ket{\zerov_n} \ket{\cv^\perp}$ with $\cv^\perp \in \CC^\perp$ of weight $u$ with probability $\frac{q^{2k}}{Z} N^\perp_u \abs{\fperp(u)}^2$ where $\fperp(u) \eqdef \reallywidehat{\pi}_{\ev}$ for an arbitrary $\ev$ of weight $u$ and $N^\perp_u$ is the number of codewords of weight $u$ in $\CC^\perp$.
\end{restatable}
The proof is given in Appendix \ref{app:mainStep2}.

\par{\bf Step 3.}
This step consists in quantifying how close to $1$ the probability of observing a dual codeword of weight in the set $\Wc$ after measuring $\psiidealQFT$ is. More specifically, we have
\begin{restatable}{proposition}{propstepthree}\label{prop:step2} 
	Under the assumptions made in Theorem \ref{th:main}, the probability of obtaining a codeword $\cv^\perp \in \CC^\perp$ of weight $u \in \Wc$ when measuring $\psiidealQFT$ is $\geq 1 - \alpha(\pi)$ for a proportion $\geq 1- \beta(\pi) $ of matrices $\Gm$, where:
	\begin{eqnarray*}
		\alpha(\pi) & \eqdef & \sum_{u \in \Wc}  \Rbra{\frac{q^{k}}{S_{u}}}^{1/4} + \sqrt{\frac{\piun}{q^{n-k}}} - 2^{-\Omega(n)}, \\
		\beta(\pi) & \eqdef & (q-1) \sum_{u \in \Wc} \sqrt{\frac{q^{k}}{S_{u}}} + \sqrt{\frac{\piun}{q^{n-k}}} + \OO{q^{-\min(k,n-k)}}.
	\end{eqnarray*}
\end{restatable}
This proposition is proved in Appendix \ref{app:mainStep3}.

\par{{\bf Step 4.}}
We first prove the following point 
\begin{lemma}\label{lem:good}
	Call $\Gc$ the set of ``good matrices''  $\Gm \in \Fq^{k \times n}$ that satisfy at the same time:
	\begin{itemize} 
		\item[$(i)$] $\varepsilon_{\Gm} \geq \varepsilon/2$ (where $\varepsilon$ and $\varepsilon_{\Gm}$ are defined in Equations \eqref{eq:epsilon} and \eqref{eq:epsG}),
		\item[$(ii)$] $Z \leq 2 q^k$.
	\end{itemize} 
	The proportion of good matrices is at least $\varepsilon/2 - \delta(\pi)$ where $\delta(\pi)\eqdef \frac{\piun}{q^{n-k}} +\OO{q^{-\min(k,n-k)}}$.
\end{lemma}
\begin{proof} By definition,
	\begin{equation*}
 		\varepsilon = \frac{1}{q^{kn}}\sum_{\Gm \in \Fq^{k\times n}} \varepsilon_{\Gm}. 
	\end{equation*}
	Let $\Bc$ be the set of matrices $\Gm$ that are not good, namely for which $(a)$ $\varepsilon_{\Gm} < \varepsilon/2$ or $(b)$ $Z > 2 q^{k}$. By Lemma \ref{lem:Z}, the density of matrices verifying $(b)$ is smaller than $\delta(\pi)$. Therefore, 
	\begin{equation*}
		\varepsilon \leq \frac{1}{q^{kn}}\sum_{\Gm \notin \Bc} 1 + \delta(\pi)\; \frac{\varepsilon}{2} \leq \frac{1}{q^{kn}}\sum_{\Gm \notin \Bc} 1 + \delta(\pi) + \frac{\varepsilon}{2}
	\end{equation*}
	which concludes the proof. 
\end{proof}
 
We use this lemma to prove that the statistical distance between the weight distributions obtained by measuring $\psiAQFT$ and $\psiidealQFT$ cannot be too far away:
\begin{lemma}\label{lem:nottofaraway}
	Let $P$, respectively $Q$, be the distribution of the weights $\abs{\cv^\perp}$ of the state $\ket{\ev} \ket{\cv^\perp}$ obtained by measuring the state $\psiAQFT$, respectively $\psiidealQFT$. We have
	\begin{equation*}
		\dstat(P,Q) \leq 1- \frac{p_t^2 \varepsilon^3}{16} + \OO{p_t^4 \varepsilon^5}+ \delta(\pi).
	\end{equation*}
\end{lemma}
\begin{proof}
	Let, 
	\begin{eqnarray*}
		P_{\Gm}(u) & \eqdef & \Prob_{\cv,\ev} \Rbra{\text{measuring $\ket{\cv^\perp}$ of weight $u$ in the $2$nd register of $\psiAQFT$ for a code choice $\Gm$}} \\
		Q_{\Gm}(u) & \eqdef & \Prob_{\cv,\ev} \Rbra{\text{measuring $\ket{\cv^\perp}$ of weight $u$ in the $2$nd register of $\psiidealQFT$ for a code choice $\Gm$}}
	\end{eqnarray*}
	We start the proof by noticing that
	\begin{eqnarray*}
		\dstat(P,Q) & = & \frac{1}{2} \sum_u \abs{P(u) - Q(u)} = \frac{1}{2} \sum_u \abs{\sum_{\Gm \in \Fq^{k \times n}} \frac{1}{q^{kn}} \Rbra{P_{\Gm}(u) - Q_{\Gm}(u)}}  \\
		& \leq & \frac{1}{q^{kn}} \sum_{\Gm \in \Fq^{k \times n}} \frac{1}{2} \sum_u \abs{P_{\Gm}(u) - Q_{\Gm}(u)} \\ 
		& = & \frac{1}{q^{kn}} \sum_{\Gm \in \Fq^{k \times n}} \dstat\Rbra{P_{\Gm}, Q_{\Gm}} \\
		& = & \sum_{\Gm \in \Gc} \frac{\dstat\Rbra{P_{\Gm}, Q_{\Gm}}}{q^{kn}} + \sum_{\Gm \notin \Gc} \frac{\dstat\Rbra{P_{\Gm}, Q_{\Gm}}}{q^{kn}} \\
		& \leq & \sum_{\Gm \in \Gc} \frac{\dtr\Rbra{\psiA, \psiideal}}{q^{kn}} + \sum_{\Gm \notin \Gc} \frac{1}{q^{kn}} \qquad \text{(by Equation \ref{eq:dstat_dtr})} \\
		& \leq & \sum_{\Gm \in \Gc} \frac{\sqrt{1 - \frac{p_t^2 \varepsilon^2}{4}}}{q^{kn}} + \sum_{\Gm \notin \Gc}  \frac{1}{q^{kn}} \qquad \text{(by Proposition \ref{prop:step1})} \\
		& \leq & \sqrt{1 - \frac{p_t^2 \varepsilon^2}{4}} \Rbra{\varepsilon/2 -\delta(\pi)} + 1 - \varepsilon/2 +\delta(\pi) \qquad \text{(by Lemma \ref{lem:good})} \\
		& \leq & \Rbra{\varepsilon/2 -\delta(\pi)} \Rbra{1- \frac{p_t^2 \varepsilon^2}{8} + \OO{p_t^4 \varepsilon^4}} + 1 - \varepsilon/2 + \delta(\pi) \\
		& \leq & 1 - \frac{p_t^2 \varepsilon^3}{16} + \OO{p_t^4 \varepsilon^5} + \delta(\pi)
	\end{eqnarray*}
which concludes the proof. 
\end{proof}

\subsection{Proof of Theorem \ref{th:main}}
	
We are now ready to prove Theorem \ref{th:main}. By Proposition \ref{prop:step2} we know that 
\begin{equation*}
	\sum_{u \in \Wc} Q(u) \geq (1-\alpha(\pi))(1-\beta(\pi)) \geq 1 - \alpha(\pi) - \beta(\pi).
\end{equation*}
But now we have the following computation,
\begin{align*}
	\sum_{u \in \Wc} P(u) & \geq \sum_{u \in \Wc} Q(u) - \dstat(P,Q) \\
	& \geq 1 - \alpha(\pi) - \beta(\pi) - 1 + \frac{p_{t}^{2}\varepsilon^{3}}{16} - \OO{p_{t}^{4}\varepsilon^{5}} - \delta(\pi) \\
	& = \frac{p_{t}^{2}\varepsilon^{3}}{16} - \OO{p_{t}^{4}\varepsilon^{5}} - \alpha(\pi) - \beta(\pi) - \delta(\pi)
\end{align*}
which concludes the proof by definition of $\alpha(\pi)$, $\beta(\pi)$ and $\delta(\pi)$.

\subsection{Application to the Hamming metric}

The assumptions of Theorem \ref{th:main} will be satisfied for the Hamming metric for weights $u$ close to $\tau^\perp n$ (where $\tau^{\perp}$ is given in Equation \eqref{def:tauPerp}) and we will prove that

\begin{theorem}\label{theo:Hamming} 
	Suppose that there exists an algorithm $\Ac$ solving with success probability $\varepsilon$ the decoding problem $\decP(q,n,k,t)$ at Hamming distance $1 \leq t \eqdef \tau n \leq (1-\delta) \dgv(n,k)$ for any arbitrary $\delta > 0$. Then, there exists a quantum algorithm which takes as input a generator matrix $\Gm \in \Fq^{k\times n}$ of a code $\CC \subseteq \Fq^n$ and outputs $\cv^{\perp} \in \CC^{\perp}$ of weight $u \in \ints{(1-\alpha)\tau^{\perp}n, (1+\alpha) \tau^{\perp}n}$ (where $\alpha$ is any arbitrary constant $> 0$) with probability over a uniform choice of $\Gm$ given by a $\Om{\frac{\varepsilon^{3}}{n} - \OO{\frac{\varepsilon^{5}}{n^{2}}} - 2^{-\Om{n}}}$ where:
	\begin{equation}\label{eq:tauPerp}
		\tau^{\perp} \eqdef \frac{1}{q} \Rbra{\sqrt{(q-1)(1-\tau)} - \sqrt{\tau}}^2.
	\end{equation}
\end{theorem}

The proof of this theorem relies on Theorem \ref{th:main}, for a suitable choice of quantum state $\ket{\pi}$. This is done by choosing $\ket{\pi} = \pitrunc$ which represents a truncated $q$-ary symmetric channel of crossover probability $\tau$. All its weights are in an interval $\ints{(1-\eta)t, (1+\eta)t}$ where $\eta$ is some positive constant which will be chosen later on. 
More precisely, let us first define the (untruncated) quantum state representing the $q$-ary symmetric channel of crossover probability $\tau$:
\begin{equation*}
	\piber \eqdef \Rbra{\sqrt{1-\tau} \ket{0} + \sqrt{\tau/(q-1)} \sumFq{\alpha} \ket{\alpha}}^{\otimes n}.
\end{equation*}
Indeed, since $\piber$ can also be written as 
\begin{equation*}
	\piber = \sumFqn{e} \prob \ket{\ev} \quad \text{ with } \quad \prob \eqdef \sqrt{(1-\tau)^{n-\abs{\ev}} \Rbra{\frac{\tau}{q-1}}^{\abs{\ev}}},
\end{equation*}
measuring $\piber$ mimics the error we have in a $q$-ary symmetric channel of crossover probability $\tau$, $\ie$,
\begin{equation*}
	\Prob(\text{measurement outputs $\ev$}) = \abs{\prob}^2 = \Rbra{\frac{\tau}{q-1}}^{\abs{\ev}} (1-\tau)^{n-\abs{\ev}}.
\end{equation*}
We will be interested in the truncated version given by
\begin{equation*}
	\pitrunc \eqdef \sum_{\substack{\ev \in \Fq^n \colon \\ \abs{\ev} \in \ints{(1-\eta)t, (1+\eta)t}}} \pitrunce \ket{\ev} \quad \text{ with } \quad \pitrunce \eqdef
	\begin{cases}
      \frac{\prob}{\sqrt{N}} & \text{if $\abs{\ev} \in \ints{(1-\eta)t, (1+\eta)t}$}\\
      0 & \text{otherwise}
    \end{cases}
\end{equation*}
where $N$ is the normalizing constant given by
\begin{equation}\label{eq:N}
	N \eqdef \sum_{\substack{\ev \in \Fq^n \\ \abs{\ev} \in \ints{(1-\eta)t, (1+\eta)t}}} |\prob|^2.
\end{equation}
It will be helpful to notice that for all $\eta >0$, $N$ is exponentially close to $1$:
\begin{lemma}\label{lemma:NHam} For all $\eta > 0$, we have 
	\begin{equation*}
		N = 1 - 2^{-\Om{n}}.
	\end{equation*}
\end{lemma}
\begin{proof}
Notice that by Equation \eqref{eq:N},
	\begin{multline} \label{eq:1mN}
		1-N = \sum_{\substack{\ev \in \Fq^n \\ \abs{\ev} \not\in \ints{(1-\eta)t, (1+\eta)t}}} \abs{\prob}^2 = \Prob_{\ev}(\abs{\ev} \notin \ints{(1-\eta)t, (1+\eta)t}) \\ = \Prob_{\ev}(\abs{\ev} < (1-\eta)\tau n) \;\; + \;\; \Prob_{\ev}(\abs{\ev} > (1+\eta)\tau n)
	\end{multline}
	where $\abs{\ev}$ is the sum of $n$ independent (binary) Bernoulli random variables of parameter $\tau$. Therefore, by Hoeffding's inequality, we have for all $\eta > 0$,   
\begin{equation*}
		\Prob_{\ev}(\abs{\ev} < (1-\eta)\tau n) \le e^{-2 \eta^2 \tau^2 n} \quad \text{and} \quad \Prob_{\ev}(\abs{\ev} > (1+\eta)\tau n) \le e^{-2 \eta^2 \tau^2 n}
	\end{equation*}
	which concludes the proof by plugging this in Equation \eqref{eq:1mN}.
\end{proof} 

Theorem \ref{theo:Hamming} is proved by showing that $\pitrunc$ satisfies all the requirements of Theorem \ref{th:main} when $\eta$ is small enough.

\subsection*{Step 1: Verification of Condition (\ref{eq:concentrationPrimal}).}
This amounts to proving the following lemma
\begin{lemma}\label{lem:piHam}
	For $\eta > 0$ small enough, we have,
	\begin{equation*}
		\frac{\braket{\pitruncp}{\unv}^2}{q^{n-k}} = 2^{-\Om{n}}.
	\end{equation*}
\end{lemma} 

Before proving this result, it will be helpful to notice that:
\begin{lemma}\label{lemma:GVHam} If $u \leq (1-\delta) \dgv(n,k)$ for some $\delta > 0$, then
	\begin{equation*}
		\frac{S_{u}}{q^{n-k}} = 2^{-\Om{n}}.
\end{equation*}
\end{lemma}
\begin{proof} Recall that the size $B_u$ of the Hamming ball of radius $u$ is of the form
	\begin{equation*}
		B_u = q^{n \; h_q(\mu)(1+o(1))}
	\end{equation*}
	where $\mu \eqdef u/n$. From this we obtain
	\begin{eqnarray*}
		\frac{S_u}{q^{n-k}} & \leq & \frac{B_u}{B_{\dgv}} \qquad \text{(since $S_u \leq B_u$ and $B_{\dgv} \leq q^{n-k}$)}\\
		& \leq & q^{n(h_q(\mu) - h_q(\drgv) + o(1))}\\
		& \leq & q^{n(h_q((1-\delta)\drgv) - h_q(\drgv) + o(1))}.
	\end{eqnarray*}
	We finish the proof by noticing that $h_q((1-\delta)\drgv) - h_q(\drgv) < 0$.
\end{proof}

We are now ready to prove Lemma \ref{lem:piHam}.

\begin{proof}[Proof of Lemma \ref{lem:piHam}]
We have the following computation,
	\begin{align}
		\braket{\pitruncp}{\unv} & = \sum_{\substack{\ev \in \Fq^n \\ \abs{\ev} \in \ints{(1-\eta)t, (1+\eta)t}}} \pitrunce \nonumber \\
		& = \sum_{r \in \ints{(1-\eta)t, (1+\eta)t}} \sum_{\substack{\ev \in \Fq^n \\ \abs{\ev} = r}} \frac{\prob}{\sqrt{N}} \nonumber \\
		& = \frac{1}{\sqrt{N}} \sum_{r \in \ints{(1-\eta)t, (1+\eta)t}} \sqrt{\binom{n}{r} (q-1)^r} \; \sqrt{\binom{n}{r}(q-1)^{r} \; (1-\tau)^{n-r} \Rbra{\frac{\tau}{q-1}}^r} \nonumber \\
		& \le \frac{1}{\sqrt{N}} \sum_{r \in \ints{(1-\eta)t, (1+\eta)t}} \sqrt{\binom{n}{r} (q-1)^r} \label{eq:BpiT} 
	\end{align}
	where in the last line we used that $\binom{n}{u}(q-1)^{u} \Rbra{\frac{\tau}{q-1}}^u (1-\tau)^{n-u} \leq 1$ for any $u \in \ints{0,n}$. Therefore, using Equation \eqref{eq:BpiT}, we have the following computation 
\begin{align*}
		\braket{\pitruncp}{\unv}^2 & \le \frac{1}{N} \Rbra{\sum_{r = \ceil{(1-\eta)t}}^{{\floor{(1+\eta)t}}} \sqrt{\binom{n}{r} (q-1)^r}} \Rbra{\sum_{r' = \ceil{(1-\eta)t}}^{{\floor{(1+\eta)t}}} \sqrt{\binom{n}{r'} (q-1)^{r'}}} \\
		& \le \frac{(n+1)^2}{N} \max\limits_{r,r' \in \ints{(1-\eta)t,(1+\eta)t}} \sqrt{\binom{n}{r} \binom{n}{r'} (q-1)^{r} (q-1)^{r'}} \quad \Rbra{\text{since $\sharp \ints{(1-\eta)t, (1+\eta)t} \le n+1$}} \\
		& \le \frac{(n+1)^2}{N} \binom{n}{(1+\eta) t} (q-1)^{(1+\eta) t} \quad \Rbra{\text{the $\max$ is reached for $r=r'= (1+\eta)t$}} \\
		& = \frac{(n+1)^2}{N} \; S_{(1+\eta)t}
\end{align*}
	Using this last inequality and Lemmas \ref{lemma:NHam} and \ref{lemma:GVHam}, we obtain
	\begin{equation*}
		\frac{\braket{\pitruncp}{\unv}^2}{q^{n-k}} = \OO{(n+1)^{2} \; \frac{S_{(1+\eta)t}}{q^{n-k}}} = 2^{-\Om{n}}
	\end{equation*}
	where we choose $\eta$ small enough such that $(1+\eta)t \leq (1-\delta')\dgv(n,k)$ for some $\delta' > 0$ (recall that by assumption $t \leq (1-\delta)\dgv(n,k)$ for $\delta > 0$).  
\end{proof}

\begin{remark}
	As explained in the introduction and Remark \ref{rem:Conditions}, Lemma \ref{lem:piHam} is not satisfied by $\piber$. Here truncating the error distribution is essential to verify this concentration lemma.
\end{remark}

\subsection*{Step 2: Verification of Conditions (\ref{eq:aboveGV}) and (\ref{eq:concentrationDual}).}
We prove here that the Conditions (\ref{eq:aboveGV}) and (\ref{eq:concentrationDual}) of Theorem \ref{th:main} are met by $\pitrunc$. This results from a combination of arguments: (i) these two conditions are met for $\piber$ (ii) $\piber$ and $\pitrunc$ are very close and so are $\piberQFT$ and $\pitruncQFT$ (because they are obtained from the first pair by applying a QFT, which is unitary).

More precisely we are going to prove that
\begin{lemma}\label{lem:NegliEtPolyHam} 
	Let $t^{\perp} \eqdef \frac{(\sqrt{(q-1)(n-t)} - \sqrt{t})^2}{q}$ and $\alpha>0$ be some constant small enough. We let $\pitruncQFT = \sumErrors \pitrunceQFT \ket{\ev}$. This state is radial and we let $\ftruncperp(w) = \pitrunceQFT$ for any element $\ev$ of $\Fq^n$ of Hamming weight $w$. 
We have
	\begin{equation*}
		p_t = \Om{\frac{1}{\sqrt{n}}},
	\end{equation*}
	\begin{equation*}
		\forall u \in \ints{t^{\perp}(1- \alpha), t^{\perp}(1+ \alpha)}, \quad \frac{q^{k}}{S_{u}} = 2^{-\Om{n}} \quad \mbox{and} 
	\sum_{u = \ceil{t^{\perp}(1- \alpha)}}^{\floor{t^{\perp}(1+ \alpha)}} S_{u} \abs{\ftruncperp(u)}^2 = 1 - 2^{-\Om{n}}.
	\end{equation*}
\end{lemma}

As explained above, to prove this result we will rely on the following lemma
\begin{lemma}\label{lemma:NegliTruncBer}
	For all $\eta > 0$, we have
	\begin{equation}
		\norm{\pitrunc - \piber} = 2^{-\Om{n}} \quad \text{and} \quad \norm{\reallywidehat{\pitrunc} - \reallywidehat{\piber}} = 2^{-\Om{n}}
	\end{equation}
\end{lemma}
\begin{proof}
	We have the following computation,
	\begin{align*}
		\norm{\pitrunc - \piber}^2 & = \sumErrors (\prob - \pitrunce)^2 \\
		& = \sum_{\substack{\ev \in \Fq^n \colon \\ \abs{\ev} \not\in \ints{(1-\eta)t, (1+\eta)t}}} \prob^2 + \sum_{\substack{\ev \in \Fq^n \colon \\ \abs{\ev} \in \ints{(1-\eta)t,(1+\eta)t}}} \Rbra{\prob - \frac{\prob}{\sqrt{N}}}^2 \\
		& = 1-N + \frac{(1-\sqrt{N})^2}{N} \sum_{\substack{\ev \in \Fq^n \\ \abs{\ev} \in \ints{(1-\eta)t, (1+\eta)t}}} \prob^2 \quad \text{ (by Equation \eqref{eq:N})} \\
		& = 1-N + (1-\sqrt{N})^2 \quad \text{(by Equation \eqref{eq:N})} \\
		& \le 2^{-\Om{n}} \quad \text{(by Lemma \eqref{lemma:NHam}).}
	\end{align*}
	The second relation follows, since the QFT is an isometry with respect to $\norm{\cdot}$.
\end{proof}

With this lemma at hand, we are ready to prove Lemma \ref{lem:NegliEtPolyHam}.
\begin{proof}[Proof of Lemma \ref{lem:NegliEtPolyHam}.]
	By definition,
	\begin{equation*}
		p_{t} = \frac{1}{N} \sum_{\ev \colon \abs{\ev} = t} (1-\tau)^{n-t} \Rbra{\frac{\tau}{q-1}}^t = \frac{\binom{n}{t}(q-1)^t \; q^{n h_{q}(\tau)}}{N} = \Om{\frac{1}{\sqrt{n}}}
	\end{equation*}
	where in the last equality we used Lemma \ref{lemma:NHam} and Stirling's formula.
	
	The equality $\frac{q^{k}}{S_{u}} = 2^{-\Om{n}}$ is verified when $u$ is sufficiently close to $t^{\perp}$ ($\alpha$ small enough), because it is readily verified that there exists some constant $\beta>0$ such that $t^\perp \geq (1+\beta) \dgv(n,n-k)$. This together with $t^\perp \leq \frac{(q-1)n}{q}$ implies that $\frac{q^{k}}{S_u} = 2^{-\Om{n}}$ for any $u$ in $\ints{t^{\perp}(1- \alpha), t^{\perp}(1+ \alpha)}$.
	
	The untruncated distribution $\piber = \sum_\ev \pi_\ev \ket{\ev}$ is radial, and so is its Fourier transform $\piberQFT = \sum_\ev \widehat{\pi}_\ev \ket{\ev}$. We let $\fperp(u) = \widehat{\pi}_\ev$ where $\ev$ is any $\ev \in \Fq^n$ of Hamming weight $u$. We notice that
	\begin{equation*}
		\sum_{u = \ceil{t^{\perp}(1- \alpha)}}^{\floor{t^{\perp}(1+ \alpha)}} S_u \abs{\fperp(u)}^2 = \sum_{\ev \in \Fq^n \colon \abs{\ev} \in \Sbra{(1-\alpha) t^\perp, (1+\alpha) t^\perp}} \abs{\widehat{\pi}_{\ev}}^{2} = \Prob_{\ev} \Rbra{\abs{\ev} \in \Sbra{(1-\alpha) t^\perp n, (1+\alpha) t^\perp n}},
	\end{equation*} 
	where 
$\tau^\perp = \frac{t^\perp}{n}$ and $\abs{\ev}$ is the sum of $n$ independent (binary) Bernoulli random variables of parameter $\tau^\perp$. Therefore, by using Hoeffding's bound again we obtain that
	\begin{equation}\label{eq:fperp}   
		\sum_{u \in [(1- \alpha) t^\perp, (1+ \alpha) t^\perp]} S_u \abs{\fperp(u)}^2 = 1 - 2^{-\Om{n}}.
	\end{equation} 
	meaning that $\fperp$ concentrates around vectors of weight $t^\perp$.
Consider the projection of $\piQFT$ and $\pitruncQFT$ on the space spanned by the states $\ket{\ev}$ for $\abs{\ev} \in \ints{(1-\alpha) t^\perp, (1+\alpha) t^\perp}$:
	\begin{eqnarray*}
		\piQFTp & \eqdef & \sum_{\ev \in \Fq^n \colon \abs{\ev} \in [(1- \alpha) t^\perp, (1+ \alpha) t^\perp]} \widehat{\pi}_\ev \ket{\ev}\\
		\pitruncQFTp & \eqdef & \sum_{\ev \in \Fq^n \colon \abs{\ev} \in [(1- \alpha) t^\perp, (1+ \alpha) t^\perp]} \pitrunceQFT \ket{\ev}
	\end{eqnarray*}
Since a projection can only reduce the norm, we have
	\begin{equation}\label{eq:distance_projection}
		\norm{\piQFTp - \pitruncQFTp} \leq \norm{\piQFT - \pitruncQFT} = 2^{-\Om{n}}.
	\end{equation}
	We deduce from the triangle inequality that 
	\begin{equation}
		\norm{\pitruncQFTp} \geq \norm{\piQFTp} - \norm{\piQFTp - \pitruncQFTp},
	\end{equation}	
	and then from Equations \eqref{eq:distance_projection} and \eqref{eq:fperp} \Big(which says $ \norm{\piQFTp}^2 = 1 - 2^{-\Om{n}}$\Big) that $\norm{\pitruncQFTp} \geq 1 - 2^{-\Om{n}}$. Since $\norm{\pitruncQFTp} \leq \norm{\pitruncQFT} = 1$ we finally obtain
	\begin{equation*}
		\norm{\pitruncQFTp} = 1 - 2^{-\Om{n}}.
	\end{equation*}
	This directly implies $\sum_{u \in \ints{(1- \alpha) t^\perp, (1+ \alpha) t^\perp}} S_u \abs{\ftruncperp(u)}^2 = \norm{\pitruncQFTp}^2 = 1 - 2^{-\Om{n}}$.
\end{proof}

\subsection*{Proof of Theorem \ref{theo:Hamming}.}

This immediately follows from Lemmas \ref{lem:piHam} and \ref{lem:NegliEtPolyHam} which show that the relevant assumptions of Theorem \ref{th:main} are verified for the
choice $\ket{\pi} = \pitrunc$ for $\eta$ small enough.

\subsection{Application to the rank metric}

The assumptions of Theorem \ref{th:main} will also be satisfied in the context of codes $\CC \subseteq \Fq^{m \times n}$ embedded with the rank metric (given some matrix in $\Fq^{m \times n}$, its weight is defined as its rank). The Gilbert-Varshamov distance $\dgv(m,n,k)$ is defined in a similar way, but it depends on three parameters here and corresponds to the largest radius $t$ of a ball in the rank metric for which 
\begin{equation*}
	q^k  B_t \leq q^{m \times n}.
\end{equation*}

We will be able to prove that
\begin{theorem}\label{theo:rank} 
	Suppose that there exists an algorithm $\Ac$ solving with success probability $\varepsilon$ the decoding problem at rank distance $1 \leq t < \dgv(m,n,k)$ where $m \geq n$. Then, there exists a quantum algorithm which takes as input a generator matrix of $\CC \subseteq \Fq^{m \times n}$ and outputs $\cv^{\perp} \in \CC^{\perp}$ of weight $u \in (t^{\perp}-\eta n, t^{\perp}]$ where $t^{\perp} \eqdef n-t$ (for any arbitrary constant $\eta > 0$) with probability over a uniform choice of the generator matrix given by a $\Om{\varepsilon^{3} - \OO{\varepsilon^{5}} - 2^{-\Om{n}}}$.
\end{theorem}

\begin{remark}
	Our assumption that $m \geq n$ can be done without loss of generality. In the case where $n > m$ we can just consider the transposed code $\transpose{\CC} \eqdef \Cbra{\transpose{\Mm} \colon \Mm \in \CC}$: taking the transpose is a linear automorphism and can be used to transform any algorithm decoding $\CC$ into an algorithm decoding $\transpose{\CC}$ with the same complexity.
\end{remark}

As in the Hamming case, Theorem \ref{theo:rank} will be a consequence of Theorem \ref{th:main}. Therefore we first have to choose appropriately a quantum state $\ket{\pi}$ that will model the noise distribution.

\subsection*{Step 1 : Choosing $\ket{\pi}$.} Let,
\begin{eqnarray}
	\ket{\pi} & \eqdef & \frac{1}{\sqrt{N}} \sum\limits_{\substack{V \le \Fq^n \\ \dim V = t } } \ket{\pi_V} \quad \text{where} \label{eq:quditRank} \\
	\ket{\pi_U} & \eqdef & \Rbra{\frac{1}{\sqrt{q^{\dim U}}} \sum_{\uv \in U} \ket{\uv}}^{\otimes m} \label{eq:piU} 
\end{eqnarray}
and $N$ is a normalizing constant. $\ket{\pi_U}$ can be viewed as a uniform superposition of matrices whose rows belong to $U$. These matrices have rank at most $\dim U$ and $\ket{\pi}$ is close to a uniform distribution of all matrices of rank $t$. We also have the following alternative description for $\ket{\pi}$.
\begin{restatable}{lemma}{lemmaF}\label{lemma:F}
	$\ket{\pi}$ is radial, $\ie$, we may write $\ket{\pi}$ as 
	\begin{equation*}
		\ket{\pi} = \sum_{\Em \in \Fq^{m \times n} \colon \abs{\Em} \leq t} \pi_{\Em} \ket{\Em} 
	\end{equation*}
	with $\ket{\Em} \eqdef \ket{\Em_{1}} \otimes \dots \otimes \ket{\Em_{m}}$ where the $\Em_i$'s denote the rows of $\Em$ and $\pi_{\Em} = f(\abs{\Em})$ where
	\begin{equation*}
		f(u) = \left\{
		\begin{array}{ll}
			\frac{\gauss{n-u}{t-u}_q}{\sqrt{q^{mt}N}} & \mbox{ if } u \leq t, \\
			0 & \mbox{ otherwise.}
		\end{array}
		\right.
	\end{equation*}
\end{restatable}

This lemma is proved in Appendix \ref{app:rankStep1}, as well as the following one, which gives estimations for $N$ and $p_t$ in order to apply Theorem \ref{th:main}.
\begin{restatable}{lemma}{lemmaNRank}\label{lemma:NRank} We have:
	\begin{equation*}
		N = \Th{\gauss{n}{t}_q} \quad \text{ and } \quad p_t = \Th{1}. 
	\end{equation*}
\end{restatable}

\subsection*{Step 2: Verification that $\mQFT{\ket{\pi}}$ is radial.}
The following proposition states that $\mQFT{\ket{\pi}}$ has actually the same form as $\ket{\pi}$ where $t$ is replaced by $n-t$:
\begin{proposition}\label{propo:FourierRank} We have,
	\begin{equation*}
		\reallywidehat{\ket{\pi}} = \frac{1}{\sqrt{N}} \sum\limits_{\substack{W \le \Fq^n \\ \dim W = n-t}} \ket{\pi_W}.
	\end{equation*}
\end{proposition}
\begin{proof} 
	We apply the QFT on $\ket{\pi}$, defined in Equation \eqref{eq:quditRank}. It gives,
	\begin{equation*}
		\reallywidehat{\ket{\pi}} = \frac{1}{\sqrt{N}} \sum\limits_{\substack{V \le \Fq^n \\ \dim V = t}} \Rbra{\frac{1}{\sqrt{q^{n+t}}} \sum_{\yv \in \Fq^n} \Rbra{\sum_{\vv \in V} \chi_{\yv}(\vv)} \ket{\yv}}^{\otimes m}
	\end{equation*}
	By distinguishing the cases where $\yv \in V^\perp$ (the dual of $V$ with the standard inner product) or not: 
	\begin{equation*} 
		\mQFT{\ket{\pi}} = \frac{1}{\sqrt{N}} \sum\limits_{\substack{V \le \Fq^n \\ \dim V = t } } \Rbra{\frac{1}{\sqrt{q^{n+t}}} \sum_{\yv \in V^\perp} q^{t} \ket{\yv}}^{\otimes m} = \frac{1}{\sqrt{N}} \sum\limits_{\substack{W \le \Fq^n \\ \dim W = n-t}} \Rbra{\frac{1}{\sqrt{q^{n-t}}} \sum_{\yv \in W} \ket{\yv}}^{\otimes m} \\
	\end{equation*} 
	which concludes the proof.  
\end{proof}

We can now straightforwardly apply Lemma \ref{lemma:F} on $\mQFT{\ket{\pi}}$ and obtain
\begin{restatable}{lemma}{lemmaPiFourierRankRadial}\label{lemma:piRank} 
	The state $\mQFT{\ket{\pi}}$ is radial and can be written as $\sum_{\Em \in \Fq^{m \times n} \colon \abs{\Em} \leq n-t} \mQFT{\pi_{\Em}} \ket{\Em}$. If we let $\fperp(u) \eqdef \mQFT{\pi_{\Em}}$ for any $\Em \in \Fq^{m \times n}$ of rank $u$, we have
	\begin{equation*}
		\fperp (u) = \left\{
		\begin{array}{ll}
			\frac{\gauss{n-u}{n-t-u}_q}{\sqrt{q^{m(n-t)}N}} & \mbox{ if } u \leq n-t, \\
			0 & \mbox{ otherwise.}
		\end{array}
		\right. 
	\end{equation*}
\end{restatable}

\subsection*{Step 3: Verification of Conditions (\ref{eq:concentrationPrimal}), (\ref{eq:aboveGV}) and (\ref{eq:concentrationDual}).}
This is achieved in the following Lemmas that are proved in Appendix \ref{app:rankStep3}.
\begin{restatable}{lemma}{lemmaGVRank}\label{lemma:GVRank}
	We have,
	\begin{equation*}
		\frac{S_t}{q^{mn-k}} = q^{-\Om{n}} \quad \mbox{and } \quad \frac{\braket{\pi}{\unv}^2}{q^{mn-k}} = q^{-\Om{n}}.
	\end{equation*}
\end{restatable} 

\begin{restatable}{lemma}{lemmaTermeNegliEtPolyRank}\label{lemma:NegliEtPolyRank} For any $\eta>0$, we have, 
\begin{equation*}
		\forall u \in \ints{(1-\eta)n-t, n-t}, \; \frac{q^k}{S_u} = q^{-\Om{n}} \quad \mbox{and} \quad \sum_{u \in \ints{(1-\eta)n-t, n-t}} S_u \abs{\fperp(u)}^2 = 1 - q^{-\Om{n}}.
	\end{equation*}
\end{restatable}

\subsection*{Proof of Theorem \ref{theo:rank}.}
This follows from Lemmas \ref{lemma:F}, \ref{lemma:piRank}, \ref{lemma:GVRank} and \ref{lemma:NegliEtPolyRank} that allow to apply Theorem \ref{th:main}, completing the proof.

 	\section{About the usefulness of our reduction.}\label{sec:useful}

It is now interesting to look at the parameters for which our reduction is useful for both the Hamming and rank metrics. 

\subsection{Hamming case}

\begin{figure}[!h]
	\centering
	\begin{subfigure}[b]{0.495\textwidth}
		\includegraphics[width=\textwidth]{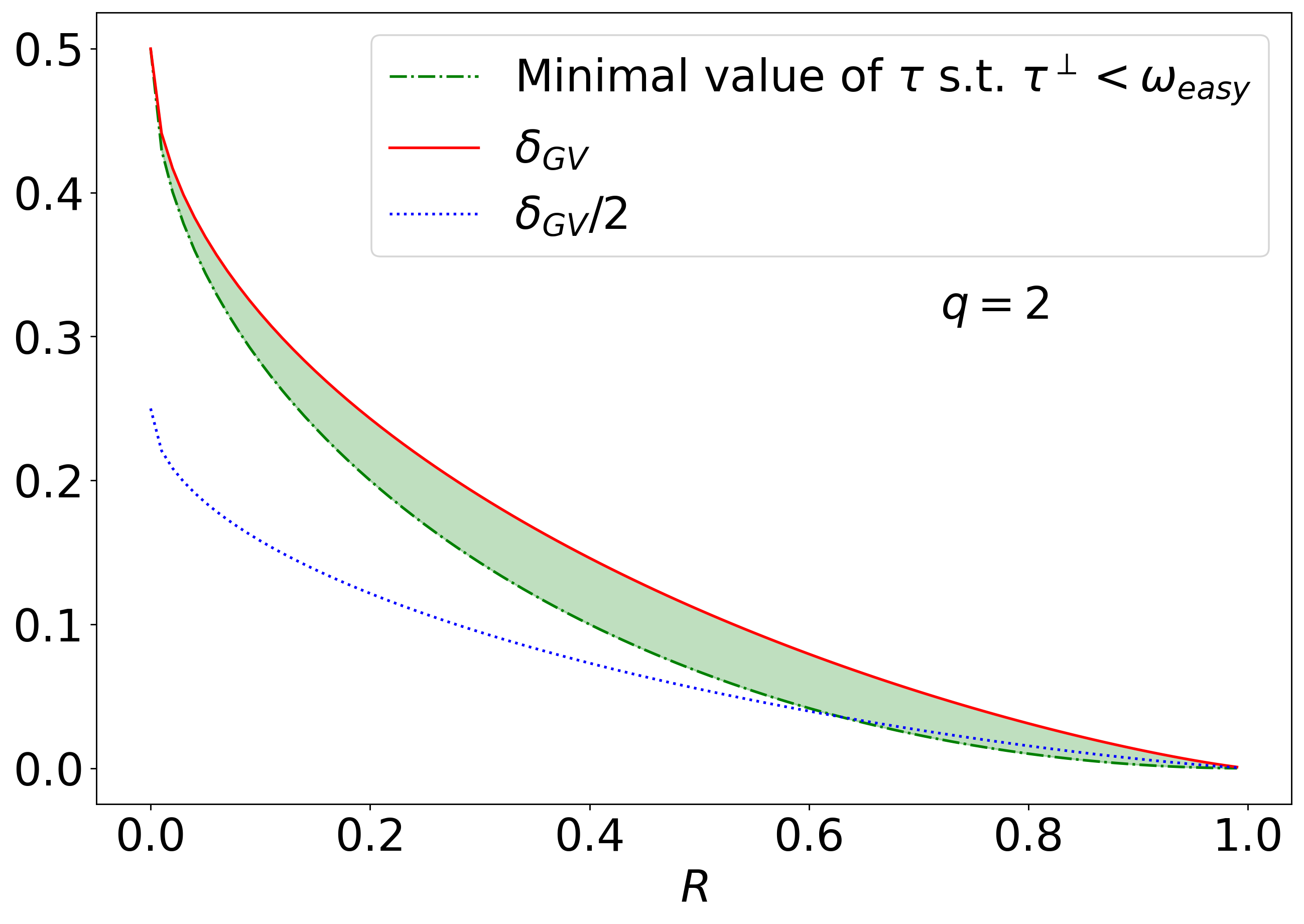}
	\end{subfigure}
	\begin{subfigure}[b]{0.495\textwidth}
		\includegraphics[width=\textwidth]{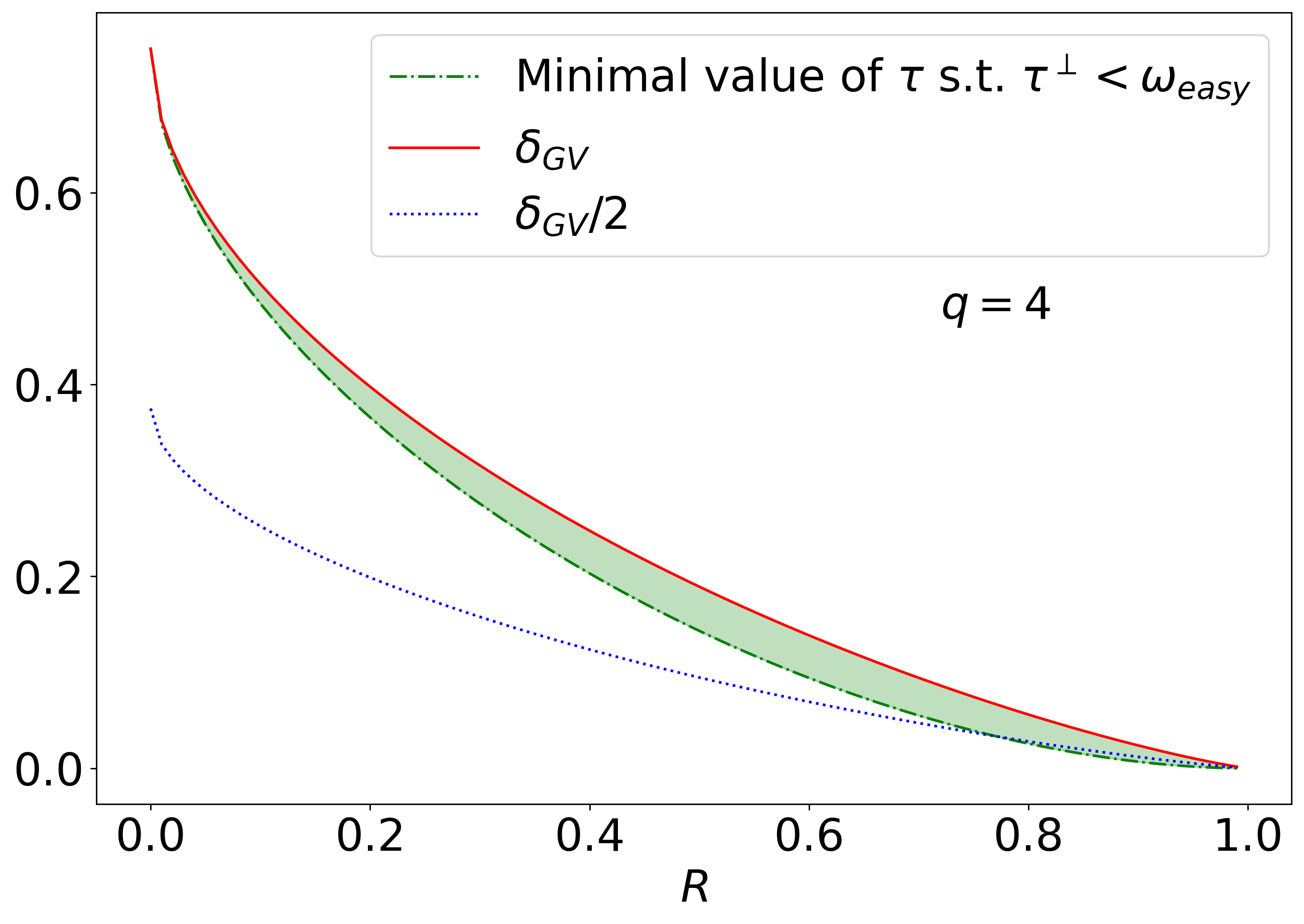}
	\end{subfigure}
	\begin{subfigure}[b]{0.495\textwidth}
		\includegraphics[width=\textwidth]{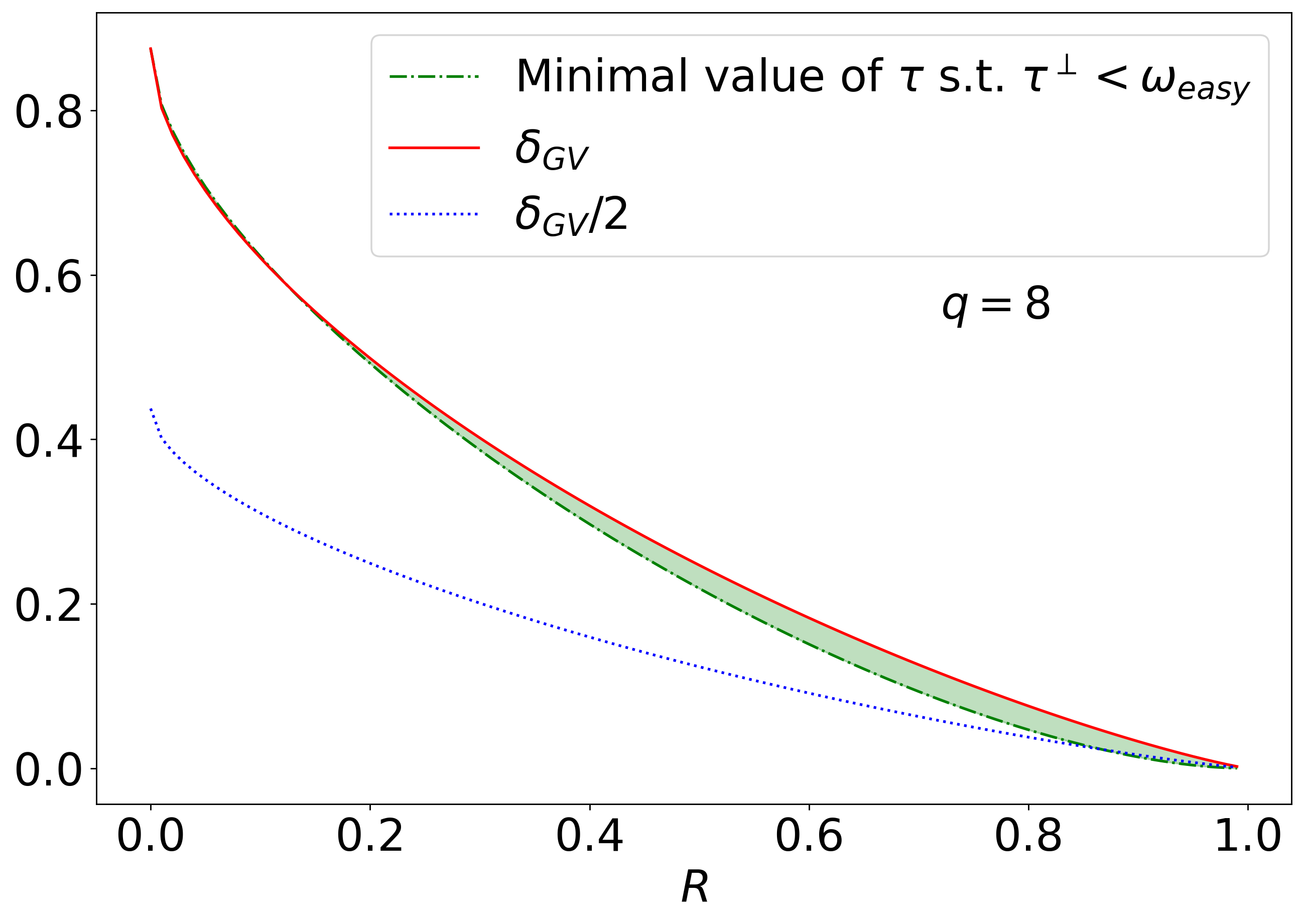}
	\end{subfigure}
	\begin{subfigure}[b]{0.495\textwidth}
		\includegraphics[width=\textwidth]{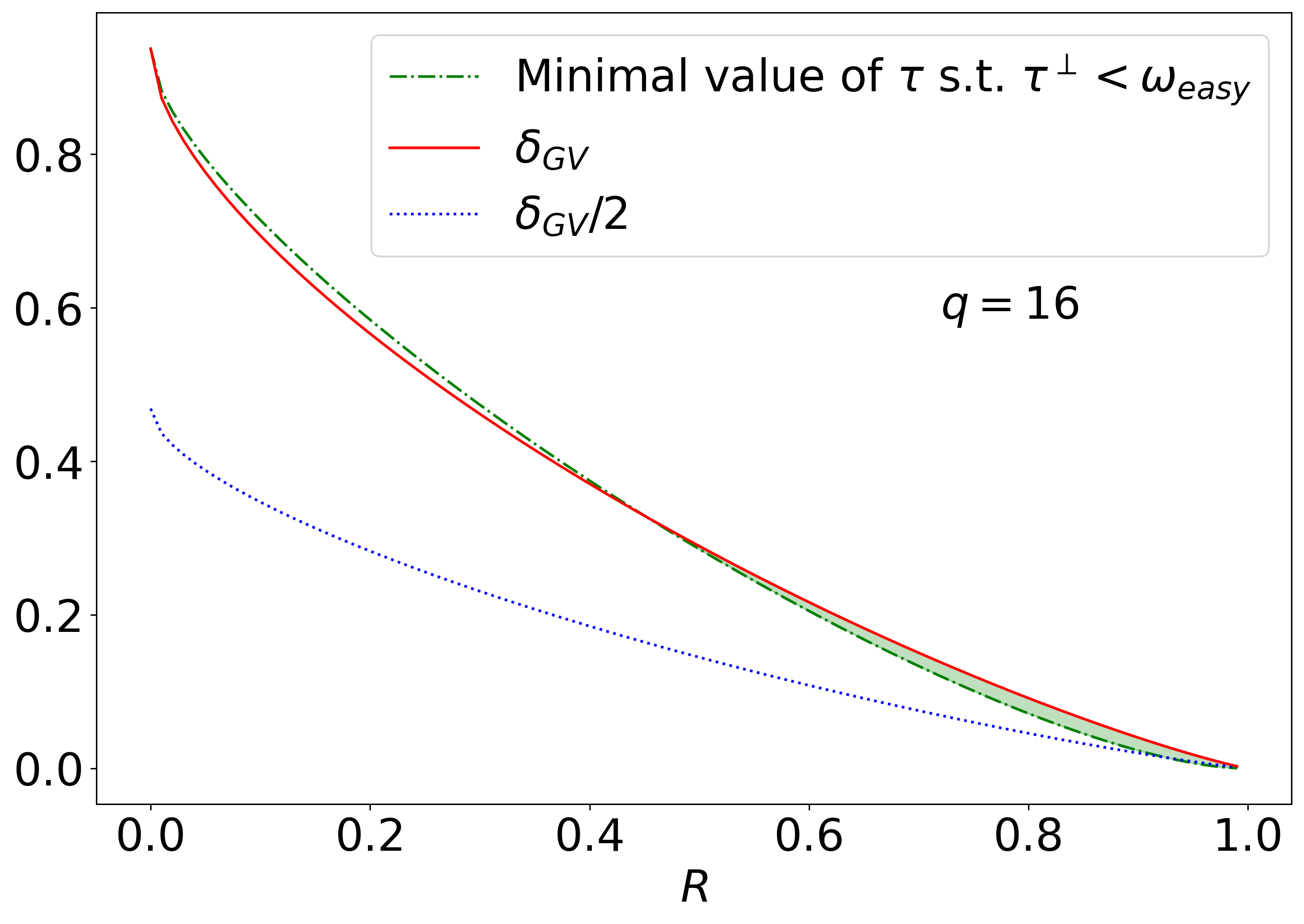}
	\end{subfigure}
	\caption{Range of values for $\tau$ as function of $R$.}
	\label{fig:RangeTau}
\end{figure}

A lower-bound on $\tau$ is obtained with the following arguments. First, if one wants to compute dual codewords (via the quantum measure) for which no poly-time algorithm is known, one has to ensure that $\tau^{\perp} < \omeasy(n,n-k) = \frac{q-1}{q} \; \frac{k}{n}$. But notice that $\tau \mapsto \tau^{\perp}$ is a decreasing involution on $\left[0,\frac{1}{2}\right]$.  Therefore, for the reduction to be meaningful, it is necessary that 
\begin{multline} \label{eq:lwBtau}
	\tau > \omeasy(n,n-k)^{\perp} = \frac{1}{q} \Rbra{\sqrt{(q-1)\Rbra{1-  \frac{q-1}{q} \; \frac{k}{n}}} - \sqrt{\frac{q-1}{q} \; \frac{k}{n}}}^{2} \\
	= \frac{q-1}{q^{2}} \Rbra{\sqrt{q- (q-1)\frac{k}{n}} - \sqrt{\frac{k}{n}}}^{2}
\end{multline} 
Furthermore, according to Theorem \ref{theo:Hamming}, the relative decoding distance $\tau$ has to verify 
\begin{equation*} 
	\tau < \frac{d_{\textup{GV}}(n,k)}{n} =  h_{q}^{-1}\Rbra{1-\frac{k}{n}} + \OO{\frac{1}{n}} = \drgv(n,k) + \OO{\frac{1}{n}}.
\end{equation*} 
Roughly speaking, it gives the tightest upper-bound for which we can expect to correctly decode with an overwhelming probability. Combining this with \eqref{eq:lwBtau} leads to a whole interval in which $\tau$ needs to lie for the reduction to work and be meaningful. 
In Figure \ref{fig:RangeTau}, we draw (asymptotically in $n$) this range of values of  $\tau$ as function of $R\eqdef \frac{k}{n}$ for different values of $q$ (the green area). 
Two important remarks can be made from these graphs:
\begin{enumerate}
	\item the range of interesting values for $\tau$ ({\em i.e.,} values such that we solve a hard instance of $\SCP$) shrinks as $q$ grows and depending on $R$;
	\item the lower bound on $\tau$ corresponding to $\tau^\perp < \omeasy(n,n-k)$ is almost always above $\drgv(n,k)/2$, meaning that in order to solve a hard instance of $\SCP$, it is a necessity that $\tau$ goes beyond the unique decoding radius.
\end{enumerate}

\subsection{Rank Case}
	
It turns out that unfortunately, for the rank metric (which is coarser than the Hamming metric), we always reduce decoding $t$ errors to finding dual codewords of weight $t^\perp = n-t$ where $t^\perp$ belongs to a range of values for which it is always easy to find codewords of this weight as we now show.
	
To verify this point, consider a linear code $\CC \subseteq \Fq^{m \times n}$ of dimension $K$ (with $m \geq n$). It is easy to find short codewords if they are above a certain range. To produce codewords of small weight, we use the fact that the dual code is a vector space of dimension $nm - K$. Thus, we can just produce codewords with $nm - K - 1$ entries equal to $0$ that will be good candidates for having a small weight by solving a linear system. The entries are chosen so as to fill columns with zeroes. It is straightforward that this strategy produces in polynomial time codewords of weight $\approx R n$ (since in our case $n \leq m$)
where $R$ is the rate of $\CC$ defined by $R \eqdef \frac{K}{mn}$.  

Notice now that $t^\perp$ is a decreasing function of the decoding distance $t$. The largest value for which we can hope to decode is the Gilbert-Varshamov distance $\dgv(m,n,K)$. The relative Gilbert-Varshamov distance $\drgv(m,n,K) \eqdef \frac{\dgv(m,n,K)}{n}$ satisfies the relation
\begin{equation*} 
	R = 1- \drgv(1+\nu-\drgv)
\end{equation*} 
where $\nu \eqdef  \frac{m}{n} \geq 1$. However, we have (where $\drgv^\perp$ is defined as $t^{\perp}/n$ when $t/n = \drgv(m,n,K)$)
\begin{equation*} 
	\frac{t^{\perp}}{n} \geq \drgv^\perp = 1 - \drgv = R + \drgv(\nu - \drgv) \geq R.
\end{equation*} 
In other words, we are always in a regime where finding codewords of relative weight $t^{\perp}/n$ is easy.
	 	\section{Concluding Remarks}

\par{\bf Considering other metrics.}
The whole approach we have followed here (properly choosing the error distribution and, if necessary, going beyond the unique decoding radius for decoding) can of course be adapted to other metrics than those we have considered here (Hamming and rank). For instance, it would be interesting to investigate it also for metrics like the Lee metric (more or less the $L_1$ norm version of the Euclidean metric on $\Z_q^n$)  which has also begun to find its way in code-based cryptography \cite{HW20} and should have a behavior closer to the Euclidean metric if the size of the alphabet grows with the code length.

\par{\bf Devising quantum algorithms for producing low weight codewords.}
Interestingly, the very same quantum reduction was recently used in \cite{CLZ22} to devise a polynomial-time algorithm for finding moderately low weight codewords for the $L_\infty$ norm in a regime of parameters for which no known polynomial time algorithms exist to perform this task. The code rate regime considered there (vanishing with the codelength) avoids the technicalities we needed here (namely truncating the error distribution). Again, this work emphasizes the role of the dual noise model obtained by applying the Fourier transform.

\par{\bf The dual error distribution.}
Moving to this quantum setting allows us to define in a natural way a dual error distribution. Indeed, consider the quantum state representing the noise added to the codeword
\begin{equation*}
	\ket{\pi} \eqdef \sumErrors \pi_{\ev} \ket{\ev}.
\end{equation*}
$(|\pi_{\ev}|^2)_{\ev \in \Fq^n} $ is a probability distribution of the error and its quantum Fourier transform $\QFT{\pi} \eqdef \sumErrors \widehat{\pi}_{\ev} \ket{\ev}$ gives a ``dual'' probability distribution of the noise $(\left|\widehat{\pi}_{\ev}\right|^2)_{\ev \in \Fq^n}$ which we view as the dual noise distribution. Note that a given noise distribution $(p_\ev)_{\ev \in \Fq^n}$ may be represented by several different quantum states $\sumErrors \pi_{\vec{e}} \ket{\ev}$ since it is sufficient that $p_\ev = |\pi_{\ev}|^2$ for all $\ev$. The choice of $\pi_\ev = \sqrt{p_{\ev}}$ is natural but not canonical and is used to define what we call the dual noise distribution. Interestingly enough, this dual noise distribution seems to capture very fundamental quantities in coding theory as we will explain in the following paragraph. Moreover, contrarily to the dual noise considered in \cite[\S IX]{F01}, which is also obtained by taking the Fourier transform and dualizing the sum-product algorithm and does not always yield a probability distribution when applied to a probability distribution, we define here a dual noise distribution that is always a probability distribution.

\par{\bf About the dual distance $t^\perp = \tau^\perp n$.}
This dual distribution also allows us to define a notion of ``dual distance''. We namely consider here probability distributions (such as the Bernoulli noise) that are concentrated around a certain weight $s$ and such that the dual distribution concentrates around another weight as well that can be considered as the ``dual weight'' $s^\perp$. In the case of the Bernoulli noise of parameter $\tau$, the dual distribution is again a Bernoulli noise distribution, of parameter $\tau^\perp$. In other words, the dual distance $t^\perp$ can be defined as $\tau^\perp n$ when $t = \tau n$. Note that the choice we made for defining the dual distribution, namely that $\pi_\ev = \sqrt{p_{\ev}}$ (where $p_\ev$ is the probability that $\ev$ is drawn from a Bernoulli distribution of parameter $\tau$), minimizes the weight $t^\perp$ around which the dual distribution converges. This gives the strongest reduction in our case.

Such an interpretation also holds in the lattice based setting (and the Euclidean distance) \cite{R05,SSTX09} when the noise is distributed over $\R^n$ according to a Gaussian distribution $D_{s}(\vec{x}) = \frac{e^{-\pi|\!|\vec{x}|\!|^{2}/s^{2}}}{s^n}$ concentrated around the distance $t = s\sqrt{\frac{n}{2\pi}}$. The dual noise is obtained by taking the Fourier transform of $\sqrt{D_s}$ and squaring the result. The dual noise is a centered Gaussian 
\begin{equation*}
t^\perp \eqdef \frac{1}{2s}\sqrt{\frac{n}{2\pi}} =\frac{n}{4 \pi t} 
\end{equation*}
where in the last equality we used that $t = s\sqrt{\frac{n}{2\pi}}$.

The same also holds in the case of the rank metric, where we define a quantum state representing the noise $\ket{\pi}$ whose rank weight concentrates around the weight $t$ we can decode. After applying the quantum Fourier transform, $\QFT{\pi}$ is a superposition of elements whose rank weight concentrates around $t^\perp \eqdef n-t$ (when we consider codes of length $n$). 
 
This notion of dual distance has an intriguing connection with a very fundamental and old issue which is still open for an overwhelming majority of metric spaces: what is the largest packing density for a given packing radius? The answer is generally not even known asymptotically. One of the most powerful technique which gives the best known bounds in the case of the Hamming metric or for the Euclidean metric over $\R^n$ is obtained through a linear programming approach \cite{D72,DL98}, see \cite{MRRW77,L79a}. Interestingly enough, the first linear programming bound in coding theory \cite{MRRW77} or the one \cite{L79a,CE03} for sphere packing in $\R^n$ can be be rephrased in terms of our dual distance and the Gilbert-Varshamov distance. Strictly speaking, the Gilbert-Varshamov is a coding theoretic notion, but it has also an analogue for $\R^n$, see \cite{DDRT22} for instance. For this purpose, let us express  the Gilbert-Varshamov distance as a function $\dgv(\delta)$ of the packing density $\delta \eqdef \frac{|\CC|}{|\EE|}$ where $|\CC|$ is the code size and $|\EE|$ the ambient space size. It is the supremum of the radius $t$ of a ball $B_t$  such that 
\begin{equation*}
	\abs{B_t} \leq \frac{1}{\delta}. 
\end{equation*}
The Gilbert-Varshamov distance for $\R^n$ and a given packing density $\delta$ can be defined similarly as the supremum of the radius $t$ of a ball $B_t$ such that  
\begin{equation*}
	\Vol(B_t) \leq \frac{1}{\delta}. 
\end{equation*}
With this notion, the first linear programming bound in coding theory \cite{MRRW77} or the one \cite{L79a} for sphere packing in $\R^n$ can be expressed as
\begin{equation}\label{eq:Hamming}
	t^\perp \geq  \dgv(\delta^*)(1+o(1)),
\end{equation}
where $\delta^*$ is the dual packing density: $\delta^* = \frac{1}{\delta}$ in the case of $\R^n$ and $\delta^* = \frac{1}{\delta |\EE|}$ in coding theoretic setting\footnote{This notion is related to the notion of dual lattice or dual code. Indeed in the code case, if the density $\frac{|\CC|}{|\EE|}$ of a linear code $\CC$ is $\delta$, the density of the dual code $\CC^\perp$ is $\delta^*$ since $\frac{|\CC^\perp|}{|\EE|}=\frac{|\CC^\perp|\cdot |\CC|}{|\EE|\cdot |\CC|}= \frac{|\EE|}{|\EE|\cdot |\CC|}=\frac{1}{\delta |\EE|}$. Similarly, for $\R^n$ if the density of a lattice $\Lambda$ is $\delta$ the density of the dual lattice $\Lambda^* \eqdef \{\yv \in \R^n: <\xv,\yv> \in \Z,\; \forall \xv \in \Lambda\}$ is $\frac{1}{\delta}=\delta^*$.}.
This indeed gives an upper-bound on $t$ since $t^\perp$ is a decreasing function of $t$. This suggests that there might be a direct proof of \eqref{eq:Hamming} relying only on the way $t^\perp$  is defined and could yield new bounds on the packing density/minimum distance of a code for the metrics where linear programming bounds are not known.
\newpage
 	\appendix
	\section{Proof of Theorem \ref{th:main}}\label{app:th:main}

\subsection{Step 1: Proof of Lemmas \ref{lem:dist_algo_ideal} and \ref{lem:Z}} \label{app:mainStep1}

Let us recall Lemma \ref{lem:dist_algo_ideal} first: \lemdistalgoideal* 

\begin{proof}
	Let $\Gc$ be the set of $(\cv,\ev)$'s that correspond to inputs of weight $t$ to $\Ac$ that are correctly decoded:
	\begin{equation*}
		\Gc  \eqdef \Cbra{(\cv,\ev) \in \CC \times \Sc_t \colon \Ac(\Gm,\cv+\ev) = \ev}. 
	\end{equation*} 
	Let us recall that
	\begin{equation*}
		\psiA = \frac{1}{\sqrt{q^k}} \sumErrors \sumCode \pi_{\ev} \ket{\ev - \Ac(\Gm,\cv+\ev)} \ket{\cv+\ev} \quad \text{and} \quad \psiideal = \frac{1}{\sqrt{Z}} \sumErrors \sumCode \pi_{\ev} \ket{\zerov_n} \ket{\cv+\ev}
	\end{equation*}
	From this we deduce by using the non-negativity of $\pi_\ev$ that
	\begin{eqnarray*}
		\braket{\psiAp}{\psiidealp} 
		& \geq & \frac{1}{\sqrt{q^k Z}} \sum_{(\cv,\ev) \in \Gc} \pi_{\ev}^2 \\
		& = & \sqrt{\frac{q^k}{Z}} S_t f(t)^2 \frac{\sharp\Gc}{q^k S_t} \\
		& = & \sqrt{\frac{q^k}{Z}} p_t \; \varepsilon_{\Gm}.
	\end{eqnarray*}
\end{proof}

\begin{remark}
	Here we do not have as in the lattice case \cite{SSTX09} to make the assumption that the decoder is ``strongly solution independent''. In our case we can indeed have a uniform superposition over all the codewords and we can just use the way our error probability is defined, namely as the ratio $ \frac{\sharp\Gc}{q^k S_t}$.
\end{remark}

All the probabilistic results of this section are easier to prove if, instead of choosing a code $\CC$ by picking uniformly at random a generator matrix $\Gm$ for it, we slightly change the probabilistic model by picking uniformly at random a parity-check matrix $\Hm \in \Fq^{(n-k) \times n}$ for it, $\ie$, 
\begin{equation*}
	\CC =\Cbra{\xv \in \Fq^n \colon \Hm \transpose{\xv} = \zerov}.
\end{equation*}
We will denote $\Prob_{\Gm}$ and $\Prob_{\Hm}$ respectively the  probabilities in the initial model and the probabilities in the new model. The two probability distributions are closely related: the first model always produces  linear codes of dimension $\leq k$ and codes of dimension $=k$ with probability $1-\OO{q^{-(n-k)}}$ whereas the second model always produces linear codes of dimension $\geq k$ and codes of dimension $=k$ with probability $1-\OO{q^{-k}}$. This relationship is expressed by the following lemma.

\begin{lemma}\label{lem:GvsH}
	Let $\Ec$ be an ensemble of linear codes of length $n$ in $\Fq$. We have
	\begin{equation*}
		\Prob_{\Gm}(\Ec) \leq \Prob_{\Hm}(\Ec) + \OO{q^{-\min(k,n-k)}}.
	\end{equation*}
\end{lemma}

With this new probabilistic model, the expected value of $Z$ is given by: \lemEZ*

\begin{proof}
	Computing $\esp_{\Hm}(Z)$ with this alternate probabilistic model is straightforward. We have
	\begin{eqnarray*}
		Z & = & \norm{\sum_{\cv \in \CC, \ev \in \Fq^n} \pi_{\ev} \ket{\zerov_n} \ket{\cv+\ev}}^2 \\
		& = & q^k \sumErrors \abs{\pi_{\ev}}^2 + \sum_{{\substack{(\cv,\ev) \neq (\cv',\ev') \colon \\ \cv+\ev=\cv'+\ev'}}} \pi_{\ev} \pi_{\ev'} \\
		& = & q^k \Rbra{1+ \sum_{\ev \neq \ev' \colon \Hm \transpose{\Rbra{\ev-\ev'}}= \zerov} \pi_{\ev} \pi_{\ev'}}
	\end{eqnarray*}
	where $\Hm$ is an arbitrary-parity check matrix for $\CC$. Let 
	\begin{equation*}
		X \eqdef \sum_{\ev \neq \ev' \colon \Hm \transpose{\Rbra{\ev-\ev'}}= \zerov} \pi_{\ev} \pi_{\ev'}.
	\end{equation*}
	The point of the probabilistic model where the parity-check matrix $\Hm$ is uniformly drawn at random is that for non-zero element $\xv \in \Fq^n$ we have
	\begin{equation*}
		\Prob_{\Hm}(\xv \in \CC) = \Prob_{\Hm}(\Hm \transpose{\xv} = \zerov) = \frac{1}{q^{n-k}}.
	\end{equation*}
	From this we deduce
	\begin{eqnarray*}
		\esp_{\Hm}(X) & = & \sum_{\ev \neq \ev'} \pi_{\ev} \pi_{\ev'} \Prob_{\Hm}((\ev - \ev') \in \CC) \\
		& = & \sum_{\ev \neq \ev'} \frac{\pi_{\ev} \pi_{\ev'}}{q^{n-k}} \\
		& \le & \sum_{\ev, \ev'} \frac{\pi_{\ev} \pi_{\ev'}}{q^{n-k}} \\
		& = & \frac{\piun}{q^{n-k}}.
	\end{eqnarray*}
	From this inequality we conclude the proof.
\end{proof}

With the help of these two lemmas, we can upper-bound the probability for $Z$ to be bigger than $q^k (1+\eta)$ for any $\eta > 0$ and prove Lemma \ref{lem:Z}, that we recall: \lemZ*

\begin{proof}
	We take back the notation from the proof of Lemma \ref{lem:EZ}. We have
	\begin{eqnarray*}
		\Prob_{\Gm}(Z > q^k (1+\eta)) & = & \Prob_{\Gm}(X > \eta) \\
		& \leq & \Prob_{\Hm}(X > \eta) + \OO{q^{-\min(k,n-k)}} \quad \text{(by Lemma \ref{lem:GvsH})} \\
		& \leq & \frac{1}{\eta}\; \esp_{\Hm}(X) + \OO{q^{-\min(k,n-k)}} \quad \text{(Markov inequality)} \\
		& \leq & \frac{1}{\eta} \; \frac{\piun}{q^{n-k}} + \OO{q^{-\min(k,n-k)}}
	\end{eqnarray*}
	which concludes the proof. 
\end{proof}

\subsection{Step 2: Proof of Lemma \ref{lem:measure}} \label{app:mainStep2}

If we apply a QFT on the second register of $\psiideal$ (given in Equation \eqref{eq:psiIdeal}), we obtain:
\begin{equation*}
	\psiidealQFT \eqdef \Rbra{\Im \otimes \textup{{QFT}} \otimes \Im} \psiideal = \frac{q^k}{\sqrt{Z}} \sumDual \reallywidehat{\pi}_{\cv^\perp} \ket{\zerov_n} \ket{\cv^\perp},
\end{equation*}
where $\QFT{\pi} = \sumErrors \widehat{\pi}_{\ev} \ket{\ev}$ is the QFT of $\ket{\pi}$. We use this remark to prove Lemma \ref{lem:measure}, that we recall: \lemmeasure* 

\begin{proof}
	For $\ev \in \Fq^n$, let 
	\begin{equation*}
		\ket{\unv_{\CC+\ev}} \eqdef \sumCode \ket{\cv+\ev}.
	\end{equation*}
	We have 
	\begin{eqnarray}
		\QFT{\unv_{\CC+\ev}} & = & \sumCode \frac{1}{\sqrt{q^n}} \sumFqn{y} \chi_{\yv}(\cv+\ev) \ket{\yv} \nonumber \\
		& = & \frac{1}{\sqrt{q^n}} \sumFqn{y} \chi_{\yv}(\ev) \sumCode \chi_{\yv}(\cv) \ket{\yv} \nonumber \\
		& = & \frac{q^k}{\sqrt{q^n}} \sumDual \chi_{\cv^\perp}(\ev) \ket{\cv^\perp} \quad \text{(since $\sumCode \chi_{\yv}(\cv) = 0$ if $\yv \notin \CC^\perp$ and $q^k$ otherwise)}\label{eq:QFT_coset}
	\end{eqnarray}
	Therefore
	\begin{eqnarray*}
		\QFT{\mathop{\sum}\limits_{\ev \in \Fq^n, \cv \in \CC} \pi_{\ev} \ket{\cv+\ev}} & = & \sumErrors \pi_{\ev} \QFT{\unv_{\CC+\ev}} \\
		& = & \frac{q^k}{\sqrt{q^n}} \sumErrors \pi_{\ev} \sumDual \chi_{\cv^\perp}(\ev) \ket{\cv^\perp} \\
		& = & q^k \sumDual \frac{1}{\sqrt{q^n}} \sumErrors \pi_{\ev} \chi_{\cv^\perp}(\ev) \ket{\cv^\perp} \\
		& = & q^k \sumDual \reallywidehat{\pi}_{\cv^\perp} \ket{\cv^\perp}.
	\end{eqnarray*}
	It follows that
	\begin{eqnarray*}
		\psiidealQFT & = & \frac{q^k}{\sqrt{Z}} \sumDual \reallywidehat{\pi}_{\cv^\perp} \ket{\zerov_n} \ket{\cv^\perp}.
	\end{eqnarray*}
	After measurement we get a state $\ket{\zerov_n} \ket{\cv^\perp}$ with probability $\frac{q^{2k}}{Z} \abs{\fperp(\abs{\cv^\perp})}^{2}$. By summing over all $\cv^\perp \in \CC^{\perp}$ of weight $u$, we conclude the proof.
\end{proof}

\subsection{Step 3 : Proof of Proposition \ref{prop:step2}} \label{app:mainStep3}

We first need for this a good estimation of $\psiidealQFT$'s amplitudes. This will be a consequence of the following lemma. 

\begin{lemma}\label{lem:Nperp}
	If the generator matrix $\Gm$ of a code $\CC$ is chosen uniformly at random in $\Fq^{k \times n}$ then the number $N^\perp_u$ of codewords of weight $u$ in $\CC^\perp$ satisfies
	\begin{equation*}
		\Prob_{\Gm}\Rbra{\abs{N^\perp_u - \frac{S_u}{q^k}} \geq  \Rbra{\frac{S_u}{q^k}}^{3/4}} \leq (q-1) \; \sqrt{\frac{q^k}{S_u}}.
	\end{equation*}
\end{lemma}
\begin{proof}\label{lemma:BT}
	Let $\und_{\xv}$ be the indicator function of the event ``$\xv \in \CC^\perp$''. By definition,
	\begin{equation}
		N_u^\perp = \sum_{\xv \in \Sc_u} \und_{\xv}.
	\end{equation} 
	We have $\esp(\und_{\xv}) = \Prob_{\Gm}(\xv \in \CC^\perp) = \frac{1}{q^k}$, implying that $\esp(N_u^\perp) = \frac{S_u}{q^k}$. 
	By using Bienaym\'{e}-Tchebychev’s inequality, we obtain:
	\begin{align}
		\Prob\Rbra{\abs{N^\perp_u -\frac {S_u}{q^k}} \geq a} & \leq \frac{\Var(N_u^\perp)}{a^2} \nonumber \\
		&= \frac{1}{a^2} \Rbra{\sum_{\xv \in \Sc_u} \Var(\und_{\xv}) + \sum_{\substack{\xv, \yv \in \Sc_u \\ \xv \neq \yv}} \esp(\und_{\xv}\und_{\yv}) - \esp(\und_{\xv})\esp(\und_{\yv})} \nonumber \\
		& \leq \frac{1}{a^2} \Rbra{\sum_{\xv \in \Sc_u} \esp(\und_{\xv}) + \sum_{\substack{\xv, \yv \in \Sc_u \\ \xv \neq \yv}} \esp(\und_{\xv}\und_{\yv}) - \esp(\und_{\xv})\esp(\und_{\yv})} \nonumber \\
		&= \frac{1}{a^{2}} \Rbra{\frac{S_u}{q^k} + \sum_{\substack{\xv, \yv \in \Sc_u \\ \xv \neq \yv}} \esp(\und_{\xv}\und_{\yv}) - \esp(\und_{\xv})\esp(\und_{\yv})} \label{eq:varD}
	\end{align}
	where we used that $\Var(\und_{\xv}) \leq \esp(\und_{\xv}^{2}) = \esp(\und_{\xv})$. Let us now upper-bound the second term of the inequality. It is readily verified that:
	\begin{equation*} 
		\esp(\und_{\xv}\und_{\yv}) = \left\{
		\begin{array}{ll}
			\nicefrac{1}{q^{k}} & \mbox{ if } \xv \mbox{ and } \yv \mbox{ are colinear}, \\
			\nicefrac{1}{q^{2k}} & \mbox{ otherwise.}
		\end{array}
		\right.
	\end{equation*} 	
	Therefore, we deduce that:
	\begin{align}
		\sum_{\substack{\xv, \yv \in \Sc_u \\ \xv \neq \yv }} \esp(\und_{\xv}\und_{\yv}) - \esp(\und_{\xv})\esp(\und_{\yv}) & = \sum_{\xv \in \Sc_u} \sum_{\substack{\yv \in \Sc_u\setminus\xv \colon \\ \text{ colinear to } \xv }} \frac{1}{q^k} - \frac{1}{q^{2k}} \nonumber\\
		&\leq \sum_{\xv \in \Sc_u} \sum_{\substack{\yv \in \Sc_u\setminus\xv \colon \\ \text{ colinear to } \xv}} \frac{1}{q^k} \nonumber \\
		&\leq  \frac{(q-2) S_u}{q^k} \label{eq:finVar}
	\end{align}
	It gives by plugging \eqref{eq:finVar} in \eqref{eq:varD}:
	\begin{equation*}
		\Prob_{\Gm}\Rbra{\abs{N^\perp_u -\frac {S_u}{q^k}} \geq a} \; \leq \; \frac{1}{a^2} \Rbra{\frac{S_u}{q^k} + \frac{(q-2)S_u}{q^k}} \; = \; \frac{(q-1)S_u}{a^2 q^k}
	\end{equation*} 
	which concludes the proof by choosing $a = \Rbra{\frac{S_u}{q^k}}^{3/4}$. 
\end{proof}

We are ready to prove Proposition \ref{prop:step2} which we now recall: \propstepthree*

\begin{proof} 
	Let $\mathcal{Q}$ be the quantum algorithm starting from $\ket{\psi}$ which computes $(\Im \otimes \textup{{QFT}} \otimes \Im)\ket{\psi}$. This algorithm succeeds when measuring a dual codeword $\cv^{\perp}\in \CC^{\perp}$ of weight $u \in \Wc$. When starting with $\psiideal$, the probability of success of $\mathcal{Q}$ is equal to $ \sum_{u \in \Wc} \frac{q^{2k} N^\perp_u}{Z} \abs{\fperp(u)}^2$ by Lemma \ref{lem:measure}. Let
	\begin{eqnarray*}
		\Bc & \eqdef & \Cbra{\Gm \in \Fq^{k \times n} \colon Z > q^k \Rbra{1 + \sqrt{\frac{\braket{\pi}{\unv}^2}{q^{n-k}}}}} \quad \text{ and } \\
		\Ec_u & \eqdef & \Cbra{\Gm \in \Fq^{k \times n} \colon \abs{N^\perp_{u} - \frac{S_u}{q^k}} \geq \Rbra{\frac{S_u}{q^k}}^{3/4}}.
	\end{eqnarray*}
	By Lemmas \ref{lem:Z} and \ref{lem:Nperp} we have that 
	\begin{equation*}
		\Prob\Rbra{\Gm \in \Bc \cup \mathop{\bigcup}\limits_{u \in \Wc} \Ec_u} \leq \beta(\pi) = (q-1)\sum\limits_{u \in \Wc} \sqrt{\frac{q^k}{S_u}} + \sqrt{\frac{\braket{\pi}{\unv}^2}{q^{n-k}}} +\OO{q^{-\min(k,n-k)}}.
	\end{equation*} 
	Therefore, for a proportion $\geq 1 - \beta(\pi)$ of codes (over matrices $\Gm$):
	\begin{enumerate}
		\item[$(i)$] $Z \leq q^k \Rbra{1+ \sqrt{\frac{\braket{\pi}{\unv}^2}{q^{n-k}}}}$ and $Z \geq q^k$ (this is true for any $\Gm$ as $\pi_{\ev} \geq 0$ for any $\ev$),
		\item[$(ii)$] for all $u$ in $\Wc$, $\abs{\frac{q^k N_u^\perp}{S_u} - 1} \leq \Rbra{\frac{q^k}{S_u}}^{1/4}$. 
	\end{enumerate}
	We deduce that for a proportion $\geq 1 - \beta(\pi)$ of codes and for all $u$ in $\Wc$:
	\begin{equation*}
		\frac{1 - \Rbra{\frac{q^k}{S_u}}^{1/4}}{1+\sqrt{\frac{\braket{\pi}{\unv}^2}{q^{n-k}}}} \leq \frac{q^k N_u^\perp}{S_u} \; \frac{q^k}{Z} \leq 1 + \Rbra{\frac{q^k}{S_u}}^{1/4}.
	\end{equation*}
	This implies that for a proportion $\geq 1 - \beta(\pi)$ of codes and for all $u$ in $\Wc$ we have that
	\begin{equation*}
		\frac{1 - \sum_{u \in \Wc} \Rbra{\frac{q^k}{S_u}}^{1/4}}{1 + \sqrt{\frac{\braket{\pi}{\unv}^2}{q^{n-k}}}} \leq \frac{q^{2k} N_u^\perp}{S_u Z} \leq 1 + \sum_{u \in \Wc} \Rbra{\frac{q^k}{S_u}}^{1/4},
	\end{equation*}
	from which we deduce that under the same conditions we also have
	\begin{multline}\label{eq:encProb}
		S_u \abs{\fperp(u)}^{2} \; \Rbra{\frac{1 - \sum_{u \in \Wc} \Rbra{\frac{q^k}{S_u}}^{1/4}}{1 + \sqrt{\frac{\braket{\pi}{\unv}^2}{q^{n-k}}}}} \\ 
		\leq \sum_{u \in \Wc} \frac{q^{2k} N^\perp_u}{Z} \abs{\fperp(u)}^2 \leq S_u \abs{\fperp(u)}^2 \Rbra{1 + \sum_{u \in \Wc} \Rbra{\frac{q^k}{S_u}}^{1/4}}. 
	\end{multline}
	Now, 
	\begin{equation*}
		\frac{1 - \sum_{u \in \Wc} \Rbra{\frac{q^k}{S_u}}^{1/4}}{1+\sqrt{\frac{\braket{\pi}{\unv}^2}{q^{n-k}}}} \geq 1 - \sum_{u \in \Wc} \Rbra{\frac{q^k}{S_u}}^{1/4} - \sqrt{\frac{\braket{\pi}{\unv}^2}{q^{n-k}}}
	\end{equation*}
	Therefore, by plugging this in Equation \eqref{eq:encProb} we have for a proportion $\geq 1- \beta(\pi)$ of codes:
	\begin{equation}
		(1-\delta) \sum_{u \in \Wc} S_u \abs{\fperp(u)}^{2} \leq 	\sum_{u \in \Wc} \frac{q^{2k} N^\perp_u}{Z} \abs{\fperp(u)}^2 \leq (1+\delta) \sum_{u \in \Wc} S_u \abs{\fperp(u)}^{2}
	\end{equation}
	where $\delta \eqdef \sum_{u \in \Wc} \Rbra{\frac{q^k}{S_u}}^{1/4} + \sqrt{\frac{\braket{\pi}{\unv}^2}{q^{n-k}}}$. 
We finish the proof by applying Lemma \ref{lem:measure}: we namely know that after measuring the state obtained by $\mathcal{Q}$, we obtain a dual codeword of weight $u$ in $\Wc$ with probability $\sum_{u \in \Wc} \frac{q^{2k} N^\perp_u}{Z} \abs{\fperp(u)}^2$.
\end{proof}

\section{Proof of Theorem \ref{theo:rank}}\label{app:theo:rank}

\subsection{Step 1 (choice of $\ket{\pi}$): Proofs of Lemmas \ref{lemma:F} and \ref{lemma:NRank}} \label{app:rankStep1}

The following lemma will be very helpful in what follows.

\begin{lemma}\cite[\S9.3, Lem. 9.3.2]{BCN89} \label{lemma:Card} 
	Let $V$ be a subspace of dimension $s$, then there are exactly $q^{(t-\ell)(s-\ell)} \gauss{n-s}{t-\ell}_q \gauss{s}{\ell}_q$ subspaces $W$ of dimension $t$ such that $\dim(V \cap W) = \ell$. 
\end{lemma}

Let us now recall Lemma \ref{lemma:F}: \lemmaF*

\begin{proof}
Let $U$ be the $\Fq$-space generated by the $\Em_{i}$'s. We denote by $u$ the dimension of $U$. We have
	\begin{equation*}
		\pi_{\Em} = \frac{1}{\sqrt{q^{mt} N}} \; \sharp \Cbra{V \le \Fq^n \colon \dim V = t \mbox{ and } U \subseteq  V } = \frac{\gauss{n-u}{t-u}_{q}}{\sqrt{q^{mt}N}},
	\end{equation*}
	where we used Lemma \ref{lemma:Card} for the last equality. It concludes the proof.
\end{proof}	

Another asymptotic expression for $N$ and an estimate for $p_t$ are given by: \lemmaNRank*

This lemma will be a consequence of the following lemmas. 	

\begin{restatable}{lemma}{lembraketpivpiw}\label{lem:braketpivpiw}
	For any $V,W \le \Fq^n$ such that $V \neq W$ and $\dim V = \dim W = t$ we have,
	\begin{equation*}
		\braket{\pi_V}{\pi_W} = q^{m \Rbra{\dim(V \cap W)-t}}.
	\end{equation*}
\end{restatable}
		
\begin{proof}
	Recall that,
	\begin{equation*}
		\ket{\pi_U} = \Rbra{\frac{1}{\sqrt{q^{\dim U}}} \sum_{\uv \in U} \ket{\uv}}^{\otimes m}
	\end{equation*}
	Therefore we have,
	\begin{equation*}
		\braket{\pi_V}{\pi_W} = \Rbra{\frac{1}{q^t} \sum_{\vv \in V} \sum_{\wv \in W} \braket{\vv}{\wv}}^m = \Rbra{\frac{1}{q^t} \sharp (V \cap W)}^m 
	\end{equation*}
which concludes the proof. 
\end{proof}

\begin{lemma}\label{lemma:N0}
	We have,
	\begin{equation*}
		\sumV \sumW \braket{\pi_V}{\pi_W} = \OO{\gauss{n}{t}_q}.
	\end{equation*}
\end{lemma}

\begin{proof}
We have
	\begin{align}
		\sumV \sumW \braket{\pi_V}{\pi_W} & = \sumV \sumW  q^{m \Rbra{\dim(V \cap W) - t}} \nonumber \\
		& = \sumV \sum_{\ell=0}^{t-1} \sum_{\substack{W \le \Fq^n \\ \dim W = t \\ \dim(W \cap V) = \ell}} \frac{1}{q^{m(t-\ell)}} \nonumber \\
		& = \sumV \sum_{\ell=0}^{t-1} \frac{1}{q^{m(t-\ell)}} \; q^{(t-\ell)^2} \gauss{t}{\ell}_q \gauss{n-t}{t-\ell}_q \label{eq:lemmaN0} \\
		& = \gauss{n}{t}_q \sum_{\ell=0}^{t-1} q^{(t-\ell-m)(t-\ell)} \gauss{t}{\ell}_q \gauss{n-t}{t-\ell}_q \nonumber
	\end{align}
	where in Equation \eqref{eq:lemmaN0} we used Lemma \ref{lemma:Card}. Now, there exists some constant $c>0$ such that:
	\begin{equation*}
		\gauss{n}{\ell}_q \le c q^{\ell(n-\ell)}.
	\end{equation*}
	Then, for some constant $C > 0$,
	\begin{align}
		\sumV \sumW \braket{\pi_V}{\pi_W}
		& \le C \gauss{n}{t}_q \sum_{\ell=0}^{t-1} q^{(t-\ell-m)(t-\ell) + \ell(t-\ell) + (t-\ell)(n-2t+\ell)} \nonumber \\
		& = C \gauss{n}{t}_q \sum_{\ell=0}^{t-1} q^{(t-\ell)(t-\ell - m + \ell + n-2t+\ell)} \nonumber \\
		& = C \gauss{n}{t}_q \sum_{\ell=0}^{t-1} q^{(t-\ell)(-t + \ell  - m + n)} \nonumber \\
		&\leq C \gauss{n}{t}_q \sum_{\ell=0}^{t-1} q^{-(t-\ell)^{2}} \quad (\text{since } n \leq m) \nonumber
	\end{align}
	which concludes the proof.
\end{proof}
		
We are now ready to prove Lemma \ref{lemma:NRank}.

\begin{proof}[Proof of Lemma \ref{lemma:NRank}] 
	By definition of $N$ we have:
	\begin{equation*}
		N = \norm{\sumV \pi_V }^{2} = \sumV \norm{\pi_V }^{2} + \sumV \sumW \Abra{\pi_{V},\pi_{W}}
	\end{equation*}
	and by definition of $\pi_V$:
	\begin{equation*}
		\sumV \norm{\pi_V }^{2} = \gauss{n}{t}_{q} \Rbra{\frac{1}{q^t} \sum_{v \in V} 1}^m = \gauss{n}{t}_{q}
	\end{equation*}
	This concludes the proof that $N = \gauss{n}{t}_{q}$ by using Lemma \ref{lemma:N0}. Now, by definition of $p_t$, we have:
	\begin{align*}
		p_t & = S_t f(t)^2 \\
			& = S_t \frac{\gauss{n-t}{0}_q^2}{q^{mt}N} \quad \text{(by Lemma \ref{lemma:F})} \\
			& = \frac{S_t}{\Th{S_t}} \quad \text{(by using the estimate for $N$ and Equation \eqref{eq:asymptRank2})} \\
	\end{align*}
	allowing us to conclude that $p_t = \Th{1}$.
\end{proof}

\subsection{Step 3: Proofs of Lemmas \ref{lemma:GVRank} and \ref{lemma:NegliEtPolyRank}} \label{app:rankStep3}

Recall Lemma \ref{lemma:GVRank} first: \lemmaGVRank*

\begin{proof}
	In order to prove the first equation, note that $\dgv(m,n,k)$ is defined such that $\frac{S_{\dgv}}{q^{mn-k}} \leq 1$. From this and Equation \eqref{eq:ratio} we deduce
	\begin{equation*}
		\frac{S_{\dgv-1}}{q^{mn-k}} = \frac{S_{\dgv-1}}{S_{\dgv}} \frac{S_{\dgv}}{q^{mn-k}} \leq \Th{q^{-(m+n-2 \dgv-1)}} = \Th{q^{-\Om{n}}}
	\end{equation*}
	where the last equality follows from the fact that (see for instance \cite{L06})
	\begin{equation}\label{eq:GVRank} 
		\dgv(n,m,k)  = \frac{m +n -\sqrt{(m-n)^2 + 4k}}{2}(1+o(1)).
	\end{equation}
	Now, for the second equation, we have
	\begin{align*}
		\frac{\braket{\pi}{\unv}^2}{q^{mn}} &= \abs{\fperp(0)}^2 \\
		&= \frac{\gauss{n}{n-t}_q^2}{q^{m(n-t)}N} \quad \mbox{(by Lemma \ref{lemma:piRank})} \\
		&= \Th{\frac{q^{mt}\gauss{n}{t}_{q}}{q^{mn}}} \quad \mbox{(By Lemma \eqref{lemma:NRank})} \\
		&= \Th{\frac{S_t}{q^{mn}}} \quad \mbox{(By Equations \eqref{eq:asymptRank1} and \eqref{eq:asymptRank2})}
	\end{align*}
	Thanks to the first equation, we complete the proof.
\end{proof}

The second lemma we need is recalled here: \lemmaTermeNegliEtPolyRank*

\begin{proof}
	The first equation can be proved in the same way as Lemma \ref{lemma:GVRank}. For the second identity, first notice (by Lemma \ref{lemma:piRank}) that
	\begin{equation*}
		\sum_{u \in \ints{n-t-\eta n, n-t}} S_u \abs{\fperp(u)}^2 = 1- \sum_{u < n-t-\eta n} S_u \abs{\fperp(u)}^2.
	\end{equation*}
	Now, we have the following computation:
	\begin{align*}
		\sum_{u < n-t-\eta n} S_u \abs{\fperp(u)}^{2} & = \sum_{u < n-t-\eta n} S_u \; \frac{\gauss{n-u}{n-u-t}_q^2}{N \; q^{m(n-t)}} \quad \mbox{(by Lemma \ref{lemma:GVRank})} \\
		&= \Th{\sum_{u < n-t-\eta n} S_u \; \frac{\gauss{n-u}{n-u-t}_q^2}{\gauss{n}{t}_{q} \; q^{m(n-t)}}} \quad \mbox{(by Lemma \ref{lemma:NRank})} \\
		&= \Th{\sum_{u < n-t-\eta n} q^{u(m+n-u)} \; \frac{q^{2(n-t-u)t}}{q^{t(n-t)} \; q^{m(n-t)}}}  \quad \mbox{(By Equations \eqref{eq:asymptRank1} and \eqref{eq:asymptRank2})} \\
		&= \Th{q^{\max_{u < n-t-\eta n}{u(m+n-u-2t)}} \; q^{(t-m)(n-t)}}
	\end{align*}
	Let $g(u) \eqdef u(m+n-u-2t)$. Then, $g'(u) = m + n - 2(t+u) \geq 0$ as $t+u \leq n \leq m$. Therefore, $g$ is an increasing function and by setting $u = n-t-\eta n$, we obtain
	\begin{align*}
		\sum_{u < n-t-\eta n} S_{u} \abs{\fperp(u)}^{2} & \leq n \; \Th{q^{(n-t-\eta n)(m-t+\eta n)} \; q^{(t-m)(n-t)}} \\
		& = n \; \Th{q^{(n-t)(t-m + m - t + \eta n)} \; q^{-\eta n(m-t+\eta n)}} \\
		& = n \; \Th{q^{-\eta n (m - t + \eta n - n + t)}} \\
		& = n \; \Th{q^{-\eta n (m -n  + \eta n )}} \\
		& = q^{-\Om{n}}
	\end{align*}
	as $n \leq m$ by assumption. It concludes the proof. 
	
\end{proof}

	\bibliographystyle{alpha}
	\addcontentsline{toc}{section}{Bibliography}
	\newcommand{\etalchar}[1]{$^{#1}$}

\end{document}